\documentclass[11pt]{article}

\usepackage[utf8]{inputenc}
\usepackage{mathtools}
\usepackage{xparse}
\usepackage{todonotes}
\usepackage[numbers,sort]{natbib}
\usepackage{algorithm}
\usepackage{algpseudocode}
\usepackage[pagebackref]{hyperref}
\usepackage{tabularx}
\usepackage{amsthm}
\usepackage{amsmath}
\usepackage{amssymb}
\usepackage{amsfonts}
\usepackage[capitalize]{cleveref}

\usepackage{fullpage}

\usepackage{authblk}

\usetikzlibrary{arrows,decorations.pathmorphing,decorations.pathreplacing,backgrounds,positioning,fit,matrix}
\usetikzlibrary{shapes,calc}

\usepackage{uri}
\newtheorem{theorem}{Theorem}
\newtheorem{lemma}[theorem]{Lemma}
\newtheorem{observation}[theorem]{Observation}
\newtheorem{corollary}[theorem]{Corollary}
\theoremstyle{definition}
\newtheorem{definition}[theorem]{Definition}

\Crefname{observation}{\text{Observation}}{\text{Observations}}

\newcommand{\TE}{\mathcal{E}}
\newcommand{\TG}{\mathcal{G}}

\newcommand{\FES}{\textsc{Feedback Edge Set}}

\newcommand{\bigO}{{O}}

\newcommand{\Ss}{\mathcal{S}}

\newcommand{\eps}{\varepsilon}

\DeclareMathOperator{\hte}{h}

\DeclareMathOperator{\absvupdate}{nrd}
\DeclareMathOperator{\s_absvupdate}{srd}

\DeclareMathOperator{\poly}{poly}
\DeclareMathOperator{\tw}{tw}

\newcommand{\STFES}{\textsc{STFES}}
\newcommand{\STFCS}{\textsc{STFCS}}

\newcommand{\tdunder}{\operatorname{td}_\downarrow}
\newcommand{\twunder}{\operatorname{tw}_\downarrow}

\DeclarePairedDelimiterX{\set}[1]{\{ }{ \} }{\setargs{#1}}
\NewDocumentCommand{\setargs}{>{\SplitArgument{1}{;}}m}
{\setargsaux#1}
\NewDocumentCommand{\setargsaux}{mm}
{\IfNoValueTF{#2}{#1} {#1\,\delimsize|\,\mathopen{}#2}}%

\DeclarePairedDelimiterX{\abs}[1]{\lvert}{\rvert}{#1}
\DeclarePairedDelimiterX{\ceil}[1]{\lceil}{\rceil}{#1}
\DeclarePairedDelimiterX{\floor}[1]{\lfloor}{\rfloor}{#1}
\DeclarePairedDelimiterX{\norm}[1]{\lVert}{\rVert}{#1}

\DeclareMathOperator{\T}{T}

\newcommand{\problemdef}[3]{
		\begin{center}
	\begin{minipage}{0.95\textwidth}
		\noindent
		\textsc{#1}
				\vspace{5pt}\\
				\setlength{\tabcolsep}{3pt}
				\begin{tabularx}{\textwidth}{@{}lX@{}}
						\textbf{Input:} 		& #2 \\
						\textbf{Question:} 	& #3
					\end{tabularx}
	\end{minipage}
		\end{center}
}

\begin{document}

\title{Feedback Edge Sets in Temporal Graphs\thanks{%
An extended abstract of this work appeared in the proceedings of
WG 2020 \cite{HMNR2020}. This version provides full proof details and corrects some errors.
Supported by the DFG, project MATE (NI 369/17).
}}

\author{Roman~Haag}
\author{Hendrik~Molter}
\author{Rolf~Niedermeier}
\author{Malte~Renken}

\affil{Algorithmics and Computational Complexity, Faculty~IV, TU Berlin, Germany\\
\{h.molter,rolf.niedermeier,m.renken\}@tu-berlin.de}

\date{ }

\maketitle              %

\begin{abstract}
The classical, linear-time solvable \FES{} problem is concerned 
with finding a minimum number of edges
intersecting all cycles in a (static, unweighted) graph.  
We provide a first study of this problem in the setting 
of \emph{temporal graphs},
where edges are present only at certain points in time.
We find that there are four natural generalizations of \FES{},
all of which turn out to be NP-hard. %
We also study the tractability of these problems with respect to several parameters
(solution size, lifetime, and number of graph vertices, among others)
and obtain some parameterized hardness 
but also fixed-parameter tractability results.
\end{abstract}

\section{Introduction}
\label{sec:intro}
A temporal graph $\TG=(V, \TE, \tau)$ has a fixed vertex set $V$ and each time-edge in~$\TE$ has a discrete time-label $t\in \{1,2, \dots, \tau \}$, where $\tau$ denotes the \emph{lifetime} of the temporal graph $\TG$.
A \emph{temporal cycle} in a temporal graph is a cycle of time-edges with increasing time-labels.
We study the computational complexity of searching for small 
\emph{feedback edge sets}, i.e., edge sets whose removal from the temporal graph destroys all temporal cycles.
We distinguish between the following two variants of feedback edge sets.
 
\begin{enumerate}
\item
\emph{Temporal feedback edge sets,} which consist of time-edges, that is, connections between two specific vertices at a specific point in time. %

\item
\emph{Temporal feedback connection sets,} which consist of vertex pairs $\{v,w\}$ causing that all time-edges between $v$ and $w$ will be removed. %
\end{enumerate}

Defining feedback edge set problems in temporal graphs is not straight-forward
because for temporal graphs the notions of paths and cycles are more involved than for static graphs. 
First, we consider two different, established 
models of temporal paths.
Temporal paths are time-respecting paths in a temporal graph.
\emph{Strict} temporal paths have strictly increasing time-labels on consecutive time-edges. %
\emph{Non-strict} temporal paths have non-decreasing time-labels on consecutive time-edges.

We focus on finding temporal feedback edge sets and temporal feedback connection sets (formalized in \Cref{sec:preliminaries}) of small cardinality in unweighted temporal graphs, each time using both the strict and non-strict temporal cycle model. We call the corresponding problems \textsc{(Strict) Temporal Feedback Edge Set} and \textsc{(Strict) Temporal Feedback Connection Set}, respectively. 

\problemdef{(Strict) Temporal Feedback Edge Set ((S)TFES)}
	{A temporal graph $\TG=(V,\TE,\tau)$ and $k \in \mathbb{N}$.}
	{Is there a (strict) temporal feedback edge set $\TE'\subseteq\TE$ of~$\TG$ with $\abs{\TE'} \leq k$?}

\problemdef{(Strict) Temporal Feedback Connection Set ((S)TFCS)}
	{A temporal graph $\TG=(V,\TE,\tau)$ and $k \in \mathbb{N}$.}
	{Is there a (strict) temporal feedback connection set~$C' \subseteq \binom{V}{2}$ of~$\TG$ with $\abs{C'} \leq~k$?}

\paragraph{Related Work.}
In static connected graphs, removing a minimum-cardinality feedback edge set 
results in a spanning tree.
This can be done in linear time via depth-first or breadth-first search.
Thus, it is natural to compare temporal feedback edge sets to the temporal analogue of a spanning tree.
This analogue is known as the \emph{minimum temporally connected (sub)graph}, which is a graph containing a time-respecting path from each vertex to every other vertex.
The concept was first introduced by \citet{KempeKK02},
and \citet{AxiotisF16} showed that in an $n$-vertex graph such a minimum
temporally connected subgraph can have $\Omega(n^2)$ edges while
\citet{CasteigtsPS19} showed that complete temporal graphs admit sparse
temporally connected subgraphs.
Additionally, \citet{AkridaGMS15} and \citet{AxiotisF16} proved that computing a minimum temporally connected subgraph is APX-hard. Considering \emph{weighted} 
temporal
graphs, there is also (partially empirical) work on computing minimum 
spanning trees,
mostly focusing on polynomial-time approximability~\cite{HFL15}.

While feedback edge sets in temporal graphs seemingly have not been 
studied before,
\citet{Agrawal16} investigated the related problem $\alpha-$\textsc{Simultaneous Feedback Edge Set}, where the edge set of a graph is partitioned into $\alpha$ color classes
and one wants to find a set of at most~$k$ edges intersecting all monochromatic cycles.
They show that this is NP-hard for $\alpha \geq 3$ colors and give a $2^{\bigO(k\alpha)} \poly(n)$-time algorithm.

Another related problem is finding $s$-$t$-separators in temporal graphs; this was studied by \citet{Berman96}, \citet{KempeKK02}, and \citet{ZschocheFMN18}.
Already here some differences were found between the strict %
and the non-strict setting,
a distinction that also matters for our results.

\paragraph{Our Contributions.}
Based on a polynomial-time many-one reduction from 3-SAT, we show NP-hardness for all four problem variants.
The properties of the corresponding construction yield more insights concerning special cases.
More specifically, the constructed graph uses $\tau=8$ distinct time-labels for the strict variants  and $\tau=3$ labels for the non-strict variants.
Similarly, we observe that our constructed graph has at most one time-edge between any pair of vertices (i.e., is \emph{simple}),
implying that the problems remain NP-hard when restricted to simple temporal graphs.
Assuming the Exponential Time Hypothesis, we can additionally prove that there is no subexponential-time algorithm solving (S)TFES or (S)TFCS.
Moreover, we show that all four problem variants are W[1]-hard when parameterized by the solution size~$k$, using a parameterized reduction from the W[1]-hard problem \textsc{Multicut in DAGs}~\cite{KratschPPW15}. 

\begin{table}[t]
\caption{Overview of our results for \textsc{(Strict) Temporal Feedback Edge Set} (marked with~*) and \textsc{(Strict) Temporal Feedback Connection Set} (marked with~**). Unmarked results apply to both variants. The parameter~$k$ denotes the solution size, $\tau$ the lifetime of the temporal graph, $L$ the maximum length of a minimal temporal cycle, and $\twunder{}$ resp.\ $\tdunder{}$ the treewidth resp.\ treedepth of the underlying graph.}
\label{tab:Results}
\Crefname{theorem}{\text{Thm.}}{\text{Thms.}}
\Crefname{corollary}{\text{Cor.}}{\text{Corss.}}
\Crefname{observation}{\text{Obs.}}{\text{Obss.}}
\centering
\begin{tabular}{l@{\hspace{0.5cm}}>{\raggedright\arraybackslash}p{6cm}@{\hspace{.5cm}}>{\raggedright\arraybackslash}p{7cm}}
\hline \\[-1em]
Param. &\multicolumn{2}{c}{Complexity}\\
\hline \\[-1em]
 &Strict variant &Non-strict variant \\
 \cline{2-3} \\[-0.5em]
none &NP-hard [\Cref{thm:NP-hardness}*/\Cref{cor:NP-hardness_SAT_STFCS}**] &NP-hard [\Cref{cor:NP-hardness_SAT_TFES_TFCS}] \\[0.5em]
$k$ &W[1]-hard [\Cref{thm:STFES_k_W1}] &W[1]-hard [\Cref{thm:STFES_k_W1}]\\[0.5em]
$\tau$ &$\tau \geq 8$: NP-h.\ [\Cref{thm:NP-hardness}*/\Cref{cor:NP-hardness_SAT_STFCS}**]&$\tau \geq 3$: NP-h.\ %
[\Cref{cor:NP-hardness_SAT_TFES_TFCS}] \\[0.5em]
$k + L$ &$\bigO(L^k \cdot \abs{\TE}^2)$ [\Cref{thm:KandCycleLengthAlg}] 
&$\bigO(L^k \cdot  \abs{\TE}^2 \log\abs{\TE})$ [\Cref{thm:KandCycleLengthAlg}]\\[0.5em]
$k + \tau$ &$\bigO(\tau^k \cdot \abs{\TE}^2)$ [\Cref{cor:StrictTautotheK}]
&open\\[0.5em]
$k + \tdunder{}$ & $2^{\bigO(\tdunder{} \cdot k)}  \cdot \abs{\TE}^2$ [\Cref{cor:treedepth}]
& $2^{\bigO(\tdunder{} \cdot k)}  \cdot \abs{\TE}^2 \log\abs{\TE}$ [\Cref{cor:treedepth}] \\[0.5em]
$|V|$ &$\bigO(2^{\abs{V}^3}\abs{V}^4 \tau)$* [\Cref{thm:FPTV}*]
$\bigO(2^{\frac{1}{2}(\abs{V}^2-\abs{V})} \cdot \abs{\TE}^2)$** [\Cref{obs:tfcs_v_fpt}**]
&$\bigO(2^{\abs{V}^3 + \abs{V}^2}\abs{V}^7\tau)$* [\Cref{thm:FPTV}*]
$\bigO(2^{\frac{1}{2}(\abs{V}^2-\abs{V})} \cdot \abs{\TE}^2 \log\abs{\TE})$** [\Cref{obs:tfcs_v_fpt}**]\\[0.5em]
$\twunder{} + \tau$ & FPT [\Cref{thm:mso}] & FPT [\Cref{thm:mso}] \\[0.5em]
\hline
\end{tabular}
\end{table}

On the positive side, based on a simple search tree, we first observe that all problem variants are fixed-parameter tractable with respect to the combined parameter~$k+L$, where~$L$ is the maximum length of a \emph{minimal} temporal cycle.
For the strict problem variants, this also implies
fixed-parameter tractability for the combined parameter~$\tau + k$.
Our main algorithmic result is to prove fixed-parameter tractability for (S)TFES 
with respect to the number of vertices~$\abs{V}$.
(For (S)TFCS, the corresponding result is straightforward as there are $\frac{1}{2}(\abs{V}^2-\abs{V})$ vertex pairs to consider.)
Finally, studying the combined parameter $\tau$ plus treewidth of the underlying graph, we show fixed-parameter tractability based on an MSO formulation.

Our results are summarized in \Cref{tab:Results}.
Notable distinctions between the different settings include the combined parameter~$k+\tau$ where the non-strict case remains open,
and the parameter~$\abs{V}$ where the proof for (S)TFES is much more involved than for (S)TFCS.

\section{Preliminaries and Basic Observations}
\label{sec:preliminaries}
We assume familiarity with standard notions from graph theory and 
(parameterized) complexity theory. %
We denote the set of positive integers with $\mathbb{N}$.
For $a \in \mathbb{N}$, we set $[a] := \{1, \dots, a\}$.
The notation~$\binom{A}{2}$ refers to all size-2 subsets of a set~$A$.
We use the following definition of temporal graphs in which the vertex set does not change with time and each time-edge has a discrete time-label~\cite{HS13,HS19,Mic16}.

\begin{definition}[Temporal Graph, Underlying Graph] An (undirected) \emph{temporal graph} $\TG=(V,\TE,\tau)$ is an ordered triple consisting of a set~$V$ of vertices, a set $\TE \subseteq {V \choose 2} \times [\tau]$ of (undirected) \emph{time-edges}, and a lifetime  $\tau \in \mathbb{N}$. 

The \emph{underlying graph} $G_\downarrow$ is the static graph obtained by removing all time-labels from $\TG$ and keeping only one edge from every multi-edge. We call a temporal graph \emph{simple} if each vertex pair is connected by at most one time-edge.
\label{def:TemporalGraph}
\end{definition}

Let $\TG=(V,\TE, \tau)$ be a temporal graph. For $i \in [\tau]$, let $E_i(\TG) :=\{ \{v,w\} \mid (\{v,w\}, i) \in \TE\}$ be the set of edges with time-label~$i$. We call the static graph $G_i(\TG) = (V,E_i(\TG)$) \emph{layer i} of $\TG$. For $t \in [\tau]$, we denote the temporal subgraph consisting of the first $t$ layers of $\TG$ by $G_{[t]}(\TG) := (V, \{ (e, i) \mid i \in [t] \land e \in E_i(\TG)\}, t)$. We omit the function parameter~$\TG$ if it is clear from the context. For some $\TE'\subseteq {V \choose 2} \times [\tau]$, we denote $\TG-\TE':=(V,\TE\setminus \TE', \tau)$.

\begin{definition}[Temporal Walk, Path, and Cycle]
	Given a temporal graph $\TG=(V,\TE,\tau)$, a \emph{temporal walk} of length $\ell$ in $\TG$ is a sequence
	$P=(e_1, e_2, \dots, e_\ell)$ of time-edges $e_i=(\{v_i,v_{i+1}\},t_i) \in \TE$
	where $t_i \leq t_{i+1}$ for all $i \in [\ell - 1]$.
	If $v_1, \dots, v_{\ell+1}$ are pairwise distinct, then $P$ is called a \emph{temporal path}.
	If $v_1, \dots, v_\ell$ are pairwise distinct, $v_{\ell+1} = v_1$, and $\ell \geq 3$, then $P$~is called a \emph{temporal cycle}.
	A temporal walk, path, or cycle is called \emph{strict} if $t_i < t_{i+1}$ for all $i \in [\ell - 1]$. 
\label{def:TemporalPath}
\end{definition}

The definitions of \textsc{(Strict) Temporal Feedback Edge Set} and \textsc{(Strict) Temporal Feedback Connection Set} (see \Cref{sec:intro}) are based on the following two sets (problem and set names are identical).

\begin{definition}[(Strict) Temporal Feedback Edge Set]
Let $\TG=(V,\TE,\tau)$ be a temporal graph.
A time-edge set $\TE' \subseteq \TE$ is called a \emph{(strict) temporal feedback edge set} of $\TG$ if $\TG'=(V,\TE \setminus \TE',\tau)$ does not contain a (strict) temporal cycle.
\label{def:TemporalFES}
\end{definition}

\begin{definition}[(Strict) Temporal Feedback Connection Set]
Let $\TG=(V,\TE,\tau)$ be a temporal graph with underlying graph $G_\downarrow=(V,E_\downarrow)$.
An edge set $C' \subseteq E_\downarrow$ is a \emph{(strict) temporal feedback connection set} of $\TG$ if $\TG'=(V,\TE',\tau)$ with $\TE'= \{(\{v,u\},t) \in \TE \mid \{v,u\} \notin C' \}$ does not contain a (strict) temporal cycle.
\label{def:TemporalFCS}
\end{definition}

The elements in a feedback connection set are known as \emph{underlying edges} (edges of $G_\downarrow$).

\paragraph{Simple Observations.}
For any given starting time and source vertex, we can compute shortest temporal paths to all other vertices in 
$\bigO(\abs{\TE} \log\abs{\TE})$ time \cite{WuCKHHW16}, respectively $\bigO(\abs{\TE})$ time for strict temporal paths \cite[Prop.~3.7]{ZschocheFMN18}.
Thus, by searching for each time-edge $(\{v, w\}, t)$ for a shortest temporal path from $w$ to $v$ (or vice versa) which starts at time $t$ and avoids the edge $\{v, w\}$,
we can record the following observation.

\begin{observation}\label{lem:shortestTour}
In $\bigO(\abs{\TE}^2 \log\abs{\TE})$ time,
we can find a shortest temporal cycle or confirm that none exists.
For the strict case $\bigO(\abs{\TE}^2)$ time suffices.
\end{observation} 

Given a shortest temporal cycle of length $L$, any temporal feedback edge or connection set
must contain an edge or connection used by that cycle.
By repeatedly searching for a shortest temporal cycle and then branching over all of its edges or connections,
we obtain the following (again the log-factor is only required in the non-strict case).

\begin{observation}\label{thm:KandCycleLengthAlg}
Let $\TG=(V,\TE,\tau)$ be a temporal graph where each temporal cycle has length at most $L \in \mathbb{N}$.
Then, (S)TFES and (S)TFCS can be solved in $\bigO(L^k \cdot \abs{\TE}^2 \log\abs{\TE})$ time.
For the strict cases $\bigO(L^k \cdot \abs{\TE}^2)$ time suffices.
\end{observation}
\begin{proof}
We can construct a simple search tree based on the fact that at least one edge from each cycle has to be in the solution.
According to \Cref{lem:shortestTour}, we can confirm that~$\TG$ is cycle-free or find some shortest cycle~$C$ in~$\bigO( \abs{\TE}^2 \cdot \log\abs{\TE})$ time
(resp.\ $\bigO(\abs{\TE}^2)$ in the strict case).
If we find a cycle~$C$, then we branch over all of its $\abs{C} \leq L$ time-edges and recursively solve the instance~$\mathcal{I}'$ remaining after removing this time-edge (underlying edge for (S)TFCS) and lowering $k$ by one.
Clearly, removing any time-edge cannot create a new temporal cycle and, thus,~$L$ is also an upper-bound for the length of a minimal temporal cycle in~$\mathcal{I}'$.
The size of the resulting search tree is upper-bounded by~$\bigO(L^k)$.  
\end{proof}
Clearly, a strict temporal cycle cannot be longer than the lifetime $\tau$.
Thus, \Cref{thm:KandCycleLengthAlg} immediately gives the following result.
\begin{corollary}
\label{cor:StrictTautotheK}
STFES and STFCS can be solved in $\bigO(\tau^k \cdot \abs{\TE}^2)$ time.
\end{corollary} 
Alternatively, we can also upper-bound $L$ in terms of the length of any cycle of the underlying graph $G_\downarrow$, which in turn can be upper-bounded by~$2^{\bigO(\tdunder{})}$~\cite[Prop.~6.2]{ND12}, 
where $\tdunder{}$ is the treedepth of the underlying graph.
\begin{corollary}\label{cor:treedepth}
\label{cor:VC}
Let $\TG$ be a temporal graph and $\tdunder{}$ be the treedepth of $G_\downarrow$.
Then, (S)TFES and (S)TFCS can be solved in $2^{\bigO(\tdunder{} \cdot k)}  \cdot \abs{\TE}^2 \log\abs{\TE}$ time.
For the strict cases $2^{\bigO(\tdunder{} \cdot k)}  \cdot \abs{\TE}^2$ time suffices.
\end{corollary}

In contrast to static graphs, $\abs{V}$ is to be considered as a useful parameter for temporal graphs because the maximum number of time-edges $\abs{\TE}$
can be arbitrarily much larger than $\abs{V}$.
However, the number of underlying edges is at most $\frac{1}{2}(\abs{V}^2 - \abs{V})$ which yields the following fixed-parameter tractability result for (S)TFCS.
\begin{observation}\label{obs:tfcs_v_fpt}
TFCS can be solved in $\mathcal{O}(2^{\frac{1}{2}(\abs{V}^2-\abs{V})} \cdot \abs{\TE}^2 \log\abs{\TE})$ time
and STFCS can be solved in $\mathcal{O}(2^{\frac{1}{2}(\abs{V}^2-\abs{V})} \cdot \abs{\TE}^2)$ time.
\end{observation}
\begin{proof}
Let $\TG$ be a temporal graph with underlying graph $G_\downarrow=(V,E_\downarrow)$.
As $G_\downarrow$ is a static graph, we have $\abs{E_\downarrow} \leq \frac{1}{2}(\abs{V}^2-\abs{V})$.
Thus, there are $2^{\abs{E_\downarrow}} \leq 2^{\frac{1}{2}(\abs{V}^2-\abs{V})}$ possible feedback connection sets,
each of which can be tested in $\bigO( \abs{\TE}^2 \log \abs{\TE})$ time resp.\ $\bigO(\abs{\TE}^2)$ time (\Cref{lem:shortestTour}).
\end{proof}

\section{Computational Hardness Results}\label{sec:hardness}
We now show that all four problem variants, (S)TFES and (S)TFCS, are NP-hard on simple temporal graphs with constant lifetime.
The proofs work by reducing from the classical 3-SAT problem. 

\begin{theorem}
STFES is NP-hard for simple temporal graphs with $\tau = 8$.
\label{thm:NP-hardness}
\end{theorem}
\begin{proof} 
We show NP-hardness via a polynomial-time many-one reduction from \textsc{3-SAT}.
For a Boolean formula~$\Phi$ in conjunctive normal form (CNF) with at most three variables per clause, \textsc{3-SAT} asks if there is a satisfying truth assignment for $\Phi$.
Let $\Phi$ be such a formula with variables $x_1, x_2, \dots, x_n$ and clauses $c_1, c_2, \dots, c_m$ of the form $c_j=(\ell_j^1 \lor \ell_j^2 \lor \ell_j^3)$.
We construct an STFES instance with temporal graph~$\TG(\Phi)$ and $k=n+2m$ as follows.

For each variable $x_i$, we introduce a variable gadget (see \Cref{fig:NP-hardness_SAT_gadgets}) with vertices~$v_i$, $v_i^T$, and $v_i^F$ and time-edges $e_{i}^T:=(\{v_i,v_i^T\},2)$, $e_{i}^F:=(\{v_i,v_i^F\},3)$, and $e_{i}^h:=(\{v_i^T,v_i^F\},1)$.
As these three edges form the temporal cycle $(e_{i}^h, e_{i}^T, e_{i}^F)$, any solution for STFES must contain at least one of them.
For each clause $c_j$, we introduce a clause gadget with four vertices, $w_j$, $w_j^1$, $w_j^2$, and $w_j^3$,
and the edges $f_{j}^{a}:=(\{w_j^1,w_j^2\},1)$, $f_{j}^{b}:=(\{w_j^2,w_j^3\},2)$, $f_{j}^1:=(\{c_j,w_j^1\},7)$, $f_{j}^2:=(\{c_j,w_j^2\},6)$, $f_{j}^3:=(\{c_j,w_j^3\},5)$ (see \Cref{fig:NP-hardness_SAT_gadgets}).
The clause gadget contains three cycles which overlap in such a way that any solution has to contain at least two out of its five edges. 

\begin{figure}[t]
  \centering
  \tikzstyle{alter}=[circle, minimum size=22pt, draw, inner sep=1pt]
    \centering
    \begin{tikzpicture}[auto, >=stealth',shorten <=1pt, shorten >=1pt]
      \tikzstyle{majarr}=[draw=black]
      \node[alter] at (0,0) (xi) {$v_i$};
      \node[alter, below left = 2 of xi] (xt) {$v_i^T$};
      \node[alter, below right = 2 of xi] (xf) {$v_i^F$};

      \draw[majarr] (xt) -- node[above left]{$e_{i}^T, 2$} (xi);
      \draw[majarr] (xi) -- node [above right]{$e_{i}^F, 3$} (xf);
      \draw[majarr] (xf) -- (xt) node [midway]{$e_{i}^h, 1$};
    \end{tikzpicture}
    \begin{tikzpicture}[auto, >=stealth',shorten <=1pt, shorten >=1pt]
      \tikzstyle{majarr}=[draw=black]
      \node[alter] at (0,0) (li_1) {$w_j^1$};
      \node[alter, right = 2 of li_1] (li_2) {$w_j^2$};
      \node[alter, right = 2 of li_2] (li_3) {$w_j^3$};
	\node[alter, below = 1.6 of li_2] (li) {$w_j$};

      \draw[majarr] (li) -- node [midway]{$f_{j}^1, 7$} (li_1);
      \draw[majarr] (li_2) -- (li) node [midway]{$f_{j}^2, 6$};
      \draw[majarr] (li_3) -- (li) node [midway]{$f_{j}^3, 5$};
      \draw[majarr] (li_1) -- (li_2) node [midway]{$f_{j}^{a}, 1$};
      \draw[majarr] (li_2) -- (li_3) node [midway]{$f_{j}^{b}, 2$};
    \end{tikzpicture}
\caption{ Variable gadget (left) and clause gadget (right) used in the proof of \Cref{thm:NP-hardness}.
	Written next to each edge are its name and time-label.} 
\label{fig:NP-hardness_SAT_gadgets}
\end{figure}

\begin{figure}[t]  
\centering
  \tikzstyle{alter}=[circle, minimum size=8pt, draw, inner sep=1pt]
    \centering
    \scalebox{1}{
    \begin{tikzpicture}[auto, >=stealth',shorten <=1pt, shorten >=1pt]
      \tikzstyle{majarr}=[draw=black]
      \node[alter] at (0,0) (s) {};
      \node[left = 1ex of s] (s_label) {$s$};
      
      \node[alter] at (3,2) (x1) {};
      \node[above right = 1pt of x1] (x1_label) {$x_1$};
      \node[alter, below left= 2ex of x1] (x1_T) {};
      \node[left = 1pt of x1_T] (x1_T_label) {\tiny{T}};
      \node[alter, below right= 2ex of x1] (x1_F) {};
      \node[right = 1pt of x1_F] (x1_F_label) {\tiny{F}};
      \draw[majarr] (x1_F) -- (x1_T) node [midway]{\tiny{1}};
      \draw[majarr] (x1_T) -- (x1) node [midway]{\tiny{2}};
      \draw[majarr] (x1) -- (x1_F) node [midway]{\tiny{3}};
      
      \node[alter, right = 55pt of x1] (x2) {};
      \node[above right = 1pt of x2] (x2_label) {$x_2$};
      \node[alter, below left= 2ex of x2] (x2_T) {};
      \node[left = 1pt of x2_T] (x2_T_label) {\tiny{T}};
      \node[alter, below right= 2ex of x2] (x2_F) {};
      \node[right = 1pt of x2_F] (x2_F_label) {\tiny{F}};
      \draw[majarr] (x2_F) -- (x2_T)  node [midway]{\tiny{1}};
      \draw[majarr] (x2_T) -- (x2) node [midway]{\tiny{2}};
      \draw[majarr] (x2) -- (x2_F) node [midway]{\tiny{3}};
      
      \node[alter, right = 55pt of x2] (x3) {};
      \node[above right = 1pt of x3] (x3_label) {$x_3$};
      \node[alter, below left= 2ex of x3] (x3_T) {};
      \node[left = 1pt of x3_T] (x3_T_label) {\tiny{T}};
      \node[alter, below right= 2ex of x3] (x3_F) {};
      \node[right = 1pt of x3_F] (x3_F_label) {\tiny{F}};
      \draw[majarr] (x3_F) -- (x3_T)  node [midway]{\tiny{1}};
      \draw[majarr] (x3_T) -- (x3) node [midway]{\tiny{2}};
      \draw[majarr] (x3) -- (x3_F) node [midway]{\tiny{3}};
      
      \node[alter, right = 55pt of x3] (x4) {};
      \node[above right = 1pt of x4] (x4_label) {$x_4$};
      \node[alter, below left= 2ex of x4] (x4_T) {};
      \node[left = 1pt of x4_T] (x4_T_label) {\tiny{T}};
      \node[alter, below right= 2ex of x4] (x4_F) {};
      \node[right = 1pt of x4_F] (x4_F_label) {\tiny{F}};
      \draw[majarr] (x4_F) -- (x4_T)  node [midway]{\tiny{1}};
      \draw[majarr] (x4_T) -- (x4) node [midway]{\tiny{2}};
      \draw[majarr] (x4) -- (x4_F) node [midway]{\tiny{3}};

      \node[alter] at (5,0) (c1_l2) {};
      \node[alter, left = 15pt of c1_l2] (c1_l1) {};
      \node[alter, right = 15pt of c1_l2] (c1_l3) {};
      \node[alter, below = 15pt of c1_l2] (c1) {};
      \draw[majarr] (c1) -- (c1_l1) node [midway]{\tiny{7}};
      \draw[majarr] (c1_l2) -- (c1) node [midway]{\tiny{6}};
      \draw[majarr] (c1_l3) -- (c1) node [midway]{\tiny{5}};
      \draw[majarr] (c1_l1) -- (c1_l2)  node [midway]{\tiny{1}};
      \draw[majarr] (c1_l2) -- (c1_l3)  node [midway]{\tiny{2}};
      
      \draw[majarr] (c1_l1) -- (x1_T) node [midway]{\tiny{4}};
      \draw[majarr] (c1_l2) -- (x2_F) node [midway]{\tiny{4}};
      \draw[majarr] (c1_l3) -- (x3_T) node [midway]{\tiny{4}};
      
      \node[below = 15pt of c1] (c1_clause) {$(x_1 \lor \neg x_2 \lor x_3)$};
      
      \node[alter, right = 120pt of c1_l1] (c2_l2) {};
      \node[alter, left = 15pt of c2_l2] (c2_l1) {};
      \node[alter, right = 15pt of c2_l2] (c2_l3) {};
      \node[alter, below = 15pt of c2_l2] (c2) {};
      \draw[majarr] (c2) -- (c2_l1) node [midway]{\tiny{7}};
      \draw[majarr] (c2_l2) -- (c2) node [midway]{\tiny{6}};
      \draw[majarr] (c2_l3) -- (c2) node [midway]{\tiny{5}};
      \draw[majarr] (c2_l1) -- (c2_l2) node [midway]{\tiny{1}};
      \draw[majarr] (c2_l2) -- (c2_l3) node [midway]{\tiny{2}};
      
      \draw[majarr] (c2_l1) -- (x2_F) node [midway]{\tiny{4}};
      \draw[majarr] (c2_l2) -- (x3_F) node [midway]{\tiny{4}};
      \draw[majarr] (c2_l3) -- (x4_T) node [midway]{\tiny{4}};

      \draw[majarr] (x1.north) to ++(0,0.4) to +(-1,0) to (s.north);
      \draw[majarr] (x2.north) to ++(0,0.6) to +(-3.5,0) to (s.north);
      \draw[majarr] (x3.north) to ++(0,0.8) to +(-6,0) to (s.north);
      \draw[majarr] (x4.north) to ++(0,1) to +(-8.5,0) to node[above = 10pt] {\tiny{1}} (s.north) ;
      
      \draw[majarr] (c1.south) to ++(0,-0.2) to +(-2,0) to (s.south);
      \draw[majarr] (c2.south) to ++(0,-0.4) to +(-6,0) to node[below = 5pt] {\tiny{8}} (s.south);

      \node[below = 15pt of c2] (c2_clause) {$(\neg x_2 \lor \neg x_3 \lor x_4)$};
      
    \end{tikzpicture}}
    \caption{Example: Reduction from 3-SAT to STFES/STFCS.}
 
\label{fig:NP-hardness_SAT}
\end{figure}

We connect clauses to variables as follows (see \Cref{fig:NP-hardness_SAT} for an example).
Let~$c_j = (\ell_j^1 \lor \ell_j^2 \lor \ell_j^3)$ be a clause of $\Phi$.
If $\ell_j^1 = x_i$, then we add the edge $(\{w_j^1,v_i^T\},4)$ and, if $\ell_j^1 = \neg x_i$, we add $(\{w_j^1,x_i^F\},4)$ (edges for $\ell_j^2$ and $\ell_j^3$ analogously).
Further, we connect a new vertex $s$ to all variable gadgets by $(\{s,v_i\},1)$ for all $i \in [n]$ and to all clause gadgets by $(\{s,w_j\},8)$ for all $j \in [m]$.
This creates three additional cycles per clause, each starting and ending in $s$.
More precisely, if $x_i$ ($\neg x_i$ is handled analogously) is the $z$-th literal of clause~$c_j$,
then $\TG(\Phi)$ contains the cycle $C_{ij}^z=((\{s,v_i\},1), (\{v_i,v_i^T\},2, (\{v_i^T,w_j^z\},4), (\{w_j^z, w_j\},8-z), (\{w_j,s\},8))$.

It is easy to see that $\TG(\Phi)$ can be computed in polynomial time.
The general idea of this reduction is to use the solution size constraint to ensure that exactly one edge from each variable gadget and exactly two edges from each clause gadget are taken.
Thus, out of the three cycles starting in $s$ and going through the clause gadget of $c_j$, only two can be disconnected by picking two edges from $\{f_j^1, f_j^2, f_j^3\}$.
The remaining cycle has to be disconnected inside its variable gadget by picking either $e_i^T$ or $e_i^F$ which ``selects'' the variable that will satisfy the clause and gives us its truth assignment.
Now we show that $(\TG(\Phi),k)$ is a yes-instance of \textsc{STFES} if and only if $\Phi$ is satisfiable.

$(\Rightarrow):$ Let $\TE'$ be a solution to the constructed STFES instance.
Due to the size constraint $k \leq n+2m$ and the cycles existing inside the gadgets,
$\TE'$ contains exactly one edge from each variable gadget and two edges from each clause gadget.
Therefore, $\TE'$ cannot contain any edges adjacent to $s$ or edges connecting variable and clause gadgets.
We obtain the solution for the 3-SAT instance by setting~$x_i$ to true if $e_{i}^T \in \TE'$ and to false if $e_{i}^F \in \TE'$ or $e_{i}^h \in \TE'$.
Assume towards contradiction that there is a clause $c_j = (\ell_j^1 \lor \ell_j^2 \lor \ell_j^3)$ which is not satisfied.
Then, in all three variable gadgets connected to $w_j$, the edge needed to go from $s$ to the corresponding literal vertex of $c_j$ is present in $\TG(\Phi) - \TE'$.
As $\TE'$ contains only two of the edges from the clause gadget, the path of one of the three literals can be extended to the vertex $w_j$ and from there back to $s$, contradicting that $\TG(\Phi) - \TE'$ is cycle-free.

$(\Leftarrow):$ For the reverse direction, suppose we have a satisfying truth assignment for $\Phi$.
We obtain a solution $\TE'=\TE_{\text{Var}} \cup \TE_{\text{Cl}}$ for the STFES instance $(\TG(\Phi)=(V,\TE,\tau=8)$, $k=n+2m)$ as follows.
For the variable gadgets, we use the variable assignment to add the feedback edges 
$$
\TE_{\text{Var}} = \{e_{i}^T \mid i \in [n], x_i = \text{true}\} \cup \{e_{i}^F \mid i \in [n], x_i = \text{false} \}.
$$
For each clause $c_j = (\ell_j^1 \lor \ell_j^2 \lor \ell_j^3)$, let $z_j \in [3]$ be the number of one of the literals satisfying the clause, i.e., $\ell_j^{z_j} = \text{true}$. We add the edges between $w_j$ and the other two literal vertices to the feedback edge set:
$$
\TE_{\text{Cl}} = \{f_{j}^z \mid j \in [m], z \in [3], z \neq z_j\}.
$$
Note that this breaks all cycles inside the variable and clause gadgets and that $\abs{\TE'}=\abs{\TE_{\text{Var}}} + \abs{\TE_{\text{Cl}}} = n + 2m$.
Cycles going through multiple gadgets but not starting in $s$ are impossible as they would use at least two edges with time-label~$4$.
It remains to show that~$\TG-\TE'$ does not contain any cycle starting and ending in $s$.
Assume towards contradiction that there is such a cycle going through the variable gadget of $x_i$ and the clause gadget of $c_j$.
Further, assume that $x_i$ was set to $\text{true}$ (the other case is handled analogously) and that, therefore, $(\{v_i,v_i^F\},3) \in \TE\setminus \TE'$.
Then, the cycle begins with $(\{s,v_i\},1), (\{v_i,v_i^F\},3), (\{v_i^F,w_j^y\},4)$ for some $y \in [3]$.
Note that the edges $e_i^h$, $f_j^a$, and $f_j^b$ cannot be used due to the time-labels.
By construction of $\TG(\Phi)$, we know that if $(\{v_i^F,w_j^y\},4)$ exists, then $\ell_j^y$ is one of the literals satisfying the clause if $x_i = \text{false}$.
Since we assumed that $x_i = \text{true}$, it holds that $y \neq z_j$ and, thus, $f_{j}^y \in \TE_{\text{Cl}}$.
It follows that there is no edge which can be appended to the temporal path and, in particular, no possibility of reaching~$s$, thus contradicting the assumption.
\end{proof}

The temporal graph constructed in the proof for \Cref{thm:NP-hardness} does not contain any pair of vertices which is connected by more than one time-edge.
Hence, each underlying edge corresponds to a single time-edge and thus the reduction implies the following corollary.

\begin{corollary}
STFCS is NP-hard even for simple temporal graphs with $\tau = 8$.
\label{cor:NP-hardness_SAT_STFCS}
\end{corollary}

A very similar reduction can also be used for the 
following.

\begin{corollary}
TFES and TFCS are both NP-hard even for simple 
temporal graphs with $\tau \geq 3$.
\label{cor:NP-hardness_SAT_TFES_TFCS}
\end{corollary}
\begin{proof}[Proof sketch]
The changes that need to be made to the reduction described in the proof of \Cref{thm:NP-hardness} are shown in \Cref{fig:NP-hardness_SAT_nonstrict}.
Here, we have to subdivide the edges between variable and clause gadgets in order to avoid cycles which go through multiple gadgets but not through $s$.
In turn, only three different time-labels are needed to create the required cycles.
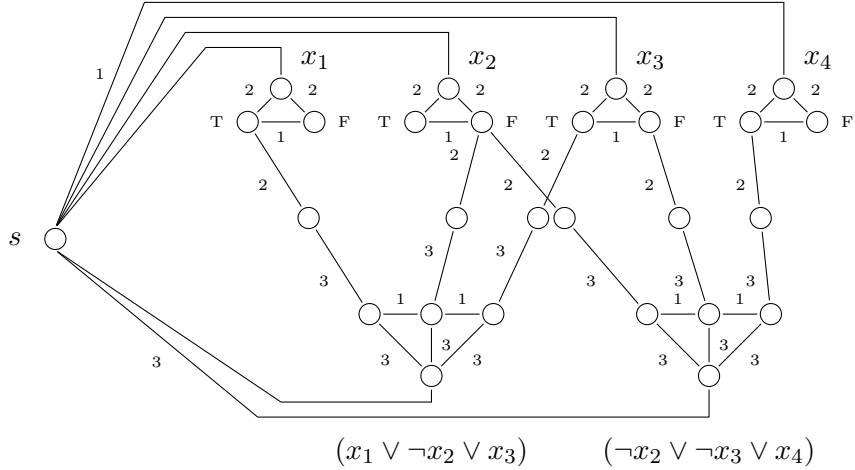
\begin{figure}[t]
\centering
  \tikzstyle{alter}=[circle, minimum size=8pt, draw, inner sep=1pt]
    \centering
    \scalebox{1}{
    \begin{tikzpicture}[auto, >=stealth',shorten <=1pt, shorten >=1pt]
      \tikzstyle{majarr}=[draw=black]
      \node[alter] at (0,0) (s) {};
      \node[left = 1ex of s] (s_label) {$s$};
      
      \node[alter] at (3,2) (x1) {};
      \node[above right = 1pt of x1] (x1_label) {$x_1$};
      \node[alter, below left= 2ex of x1] (x1_T) {};
      \node[left = 1pt of x1_T] (x1_T_label) {\tiny{T}};
      \node[alter, below right= 2ex of x1] (x1_F) {};
      \node[right = 1pt of x1_F] (x1_F_label) {\tiny{F}};
      \draw[majarr] (x1_F) -- (x1_T) node [midway]{\tiny{1}};
      \draw[majarr] (x1_T) -- (x1) node [midway]{\tiny{2}};
      \draw[majarr] (x1) -- (x1_F) node [midway]{\tiny{2}};
      
      \node[alter, right = 55pt of x1] (x2) {};
      \node[above right = 1pt of x2] (x2_label) {$x_2$};
      \node[alter, below left= 2ex of x2] (x2_T) {};
      \node[left = 1pt of x2_T] (x2_T_label) {\tiny{T}};
      \node[alter, below right= 2ex of x2] (x2_F) {};
      \node[right = 1pt of x2_F] (x2_F_label) {\tiny{F}};
      \draw[majarr] (x2_F) -- (x2_T) node [midway]{\tiny{1}};
      \draw[majarr] (x2_T) -- (x2) node [midway]{\tiny{2}};
      \draw[majarr] (x2) -- (x2_F) node [midway]{\tiny{2}};
      
      \node[alter, right = 55pt of x2] (x3) {};
      \node[above right = 1pt of x3] (x3_label) {$x_3$};
      \node[alter, below left= 2ex of x3] (x3_T) {};
      \node[left = 1pt of x3_T] (x3_T_label) {\tiny{T}};
      \node[alter, below right= 2ex of x3] (x3_F) {};
      \node[right = 1pt of x3_F] (x3_F_label) {\tiny{F}};
      \draw[majarr] (x3_F) -- (x3_T) node [midway]{\tiny{1}};
      \draw[majarr] (x3_T) -- (x3) node [midway]{\tiny{2}};
      \draw[majarr] (x3) -- (x3_F) node [midway]{\tiny{2}};
      
      \node[alter, right = 55pt of x3] (x4) {};
      \node[above right = 1pt of x4] (x4_label) {$x_4$};
      \node[alter, below left= 2ex of x4] (x4_T) {};
      \node[left = 1pt of x4_T] (x4_T_label) {\tiny{T}};
      \node[alter, below right= 2ex of x4] (x4_F) {};
      \node[right = 1pt of x4_F] (x4_F_label) {\tiny{F}};
      \draw[majarr] (x4_F) -- (x4_T) node [midway]{\tiny{1}};
      \draw[majarr] (x4_T) -- (x4) node [midway]{\tiny{2}};
      \draw[majarr] (x4) -- (x4_F) node [midway]{\tiny{2}};

      \node[alter] at (5,-1) (c1_l2) {};
      \node[alter, left = 15pt of c1_l2] (c1_l1) {};
      \node[alter, right = 15pt of c1_l2] (c1_l3) {};
      \node[alter, below = 15pt of c1_l2] (c1) {};
      \draw[majarr] (c1) -- (c1_l1) node [midway]{\tiny{3}};
      \draw[majarr] (c1_l2) -- (c1) node [midway]{\tiny{3}};
      \draw[majarr] (c1_l3) -- (c1) node [midway]{\tiny{3}};
      \draw[majarr] (c1_l1) -- (c1_l2) node [midway]{\tiny{1}};
      \draw[majarr] (c1_l2) -- (c1_l3) node [midway]{\tiny{1}};
      
		\node[alter] at ($(c1_l1)!0.5!(x1_T)$) (l1_mid) {};
		\draw[majarr] (l1_mid) -- (x1_T) node [midway]{\tiny{2}};
		\draw[majarr] (c1_l1) -- (l1_mid) node [midway]{\tiny{3}};
		
		\node[alter] at ($(c1_l2)!0.5!(x2_F)$) (l2_mid) {};
		\draw[majarr] (l2_mid) -- (x2_F) node [midway]{\tiny{2}};
		\draw[majarr] (c1_l2) -- (l2_mid) node [midway]{\tiny{3}};

		\node[alter] at ($(c1_l3)!0.5!(x3_T)$) (l3_mid) {};
		\draw[majarr] (l3_mid) -- (x3_T) node [midway]{\tiny{2}};
		\draw[majarr] (c1_l3) -- (l3_mid) node [midway]{\tiny{3}};

      \node[below = 15pt of c1] (c1_clause) {$(x_1 \lor \neg x_2 \lor x_3)$};
      
      \node[alter, right = 120pt of c1_l1] (c2_l2) {};
      \node[alter, left = 15pt of c2_l2] (c2_l1) {};
      \node[alter, right = 15pt of c2_l2] (c2_l3) {};
      \node[alter, below = 15pt of c2_l2] (c2) {};
      
      \draw[majarr] (c2) -- (c2_l1) node [midway]{\tiny{3}};
      \draw[majarr] (c2_l2) -- (c2) node [midway]{\tiny{3}};
      \draw[majarr] (c2_l3) -- (c2) node [midway]{\tiny{3}};
      \draw[majarr] (c2_l1) -- (c2_l2) node [midway]{\tiny{1}};
      \draw[majarr] (c2_l2) -- (c2_l3) node [midway]{\tiny{1}};
      
      \node[alter] at ($(c2_l1)!0.5!(x2_F)$) (l1_mid_2) {};      
      \node[alter] at ($(c2_l2)!0.5!(x3_F)$) (l2_mid_2) {};      
      \node[alter] at ($(c2_l3)!0.5!(x4_T)$) (l3_mid_2) {};
      
      \draw[majarr] (l1_mid_2) -- (x2_F) node [midway]{\tiny{2}};
      \draw[majarr] (l2_mid_2) -- (x3_F) node [midway]{\tiny{2}};
      \draw[majarr] (l3_mid_2) -- (x4_T) node [midway]{\tiny{2}};
      
      \draw[majarr] (c2_l1) -- (l1_mid_2) node [midway]{\tiny{3}};
      \draw[majarr] (c2_l2) -- (l2_mid_2) node [midway]{\tiny{3}};
      \draw[majarr] (c2_l3) -- (l3_mid_2) node [midway]{\tiny{3}};

      \draw[majarr] (x1.north) to ++(0,0.4) to +(-1,0) to (s.north);
      \draw[majarr] (x2.north) to ++(0,0.6) to +(-3.5,0) to (s.north);
      \draw[majarr] (x3.north) to ++(0,0.8) to +(-6,0) to (s.north);
      \draw[majarr] (x4.north) to ++(0,1) to +(-8.5,0) to node[above = 10pt] {\tiny{1}} (s.north) ;
      
      \draw[majarr] (c1.south) to ++(0,-0.2) to +(-2,0) to (s.south);
      \draw[majarr] (c2.south) to ++(0,-0.4) to +(-6,0) to node[below = 5pt] {\tiny{3}} (s.south);

      \node[below = 15pt of c2] (c2_clause) {$(\neg x_2 \lor \neg x_3 \lor x_4)$};
      
    \end{tikzpicture}}
    \caption{Example: Reduction from 3-SAT to TFES/TFCS.}
 
\label{fig:NP-hardness_SAT_nonstrict}
\end{figure}
\end{proof}

We can also observe that the strict problem variants are NP-hard even if all edges are present at all times.
This problem is essentially equivalent to selecting a set of edges of the underlying graph that intersects all cycles of length at most $\tau$, which is known to be NP-hard~\cite[Thm.~1]{yannakakis1981edge}.

\begin{observation}\label{obs:corehard}
	\STFES{} and \STFCS{} are NP-hard even on temporal graphs where all edges are present at all times,
	even with $\tau = 3$, planar underlying graph $G_\downarrow$, and $\Delta(G_\downarrow) = 7$.
\end{observation}
\begin{proof}
	Let $\TG = (V, \TE, 3)$ be a temporal graph with three layers $E_1 = E_2 = E_3$.
	Then an edge set $C' \subseteq E_\downarrow$ is a strict temporal feedback connection set 
	if and only if it intersects all triangles in $G_\downarrow$.
	
	Similarly, it is easy to see that a strict temporal feedback edge set $\TE' \subseteq \TE$ 
	must contain at least three time-edges from every triangle in $G_\downarrow$
	and exactly three time-edges from each triangle suffice.
	
	Since it is NP-hard to determine whether there exists a set of $k$ edges intersecting all triangles
	in a planar graph with maximum degree $7$ \cite[Thm.~2.1]{BRUGMANN200951}, the claim follows.
\end{proof}

We next show that our problems are W[1]-hard when parameterized by the solution size $k$ with a parameterized reduction from \textsc{Multicut in DAGs}~\cite{KratschPPW15}.
The idea here is that we can simulate a DAG $D$ by an undirected temporal graph 
by first subdividing all edges of $D$ and then assigning time-labels according to a topological ordering.
This ensures that each path in $D$ corresponds to a path in the resulting temporal graph and vice versa.
By adding a reverse edge from $t$ to $s$ for each terminal pair $(s, t)$ of the \textsc{Multicut} instance,
each $s$-$t$-path in~$D$ corresponds to a temporal cycle in the temporal graph and vice versa.

\begin{theorem}
\textsc{(S)TFES} and \textsc{(S)TFCS}, parameterized by the solution size~$k$, are W[1]-hard.
\label{thm:STFES_k_W1}
\end{theorem}
We prove W[1]-hardness with a parameterized reduction from \textsc{Multicut in DAGs} parameterized by the solution size.

\problemdef{Multicut in DAGs}
	{A DAG $D=(V,A)$, a set of terminal pairs $\mathcal{T}=\{(s_i,t_i)\mid i \in [r] \text{ and } s_i,t_i \in V\}$, and an integer $k$.}
	{Is there a cut-set $Z\subseteq V$ of at most $k$ nonterminal vertices of~$G$
		 such that for all $i \in [r]$ the terminal $t_i$ is not reachable from $s_i$ in $G-Z$?}

\textsc{Multicut in DAGs} was shown to be W[1]-hard when parameterized by~$k$ by \citet[Thm.~1.2]{KratschPPW15} who also provided the following lemma which will simplify our proof by further restricting the input instance.

\begin{lemma}[{\citet[Lemma~2.1]{KratschPPW15}}]
Given a \textsc{Multicut in DAGs} instance $(D, \mathcal{T}, k)$ with $D=(V,A)$,
one can compute in polynomial time an equivalent instance $(D', \mathcal{T}', k')$ with $D'=(V',A')$ such that
\begin{enumerate}
\item $\abs{\mathcal{T}}=\abs{\mathcal{T}'}$ and $k=k'$;
\item $\mathcal{T}'=\{(s_i',t_i') \mid i \in [r]\}$ and all terminals $s_i'$ and $t_i'$ are pairwise distinct; and
\item for each $v \in V$ and $i \in [r]$ we have $(v, s_i')\notin A'$ and $(t_i', v)\notin A'$.
\end{enumerate}
\label{lem:W1_KratschPPW}
\end{lemma}
We will from now on assume that we have an instance with these properties. 
The goal of our reduction will be to create one temporal cycle for each terminal pair.
Since there is a temporal path from $s_i$ to $t_i$ (the pair can be ignored otherwise), we can create a cycle by adding a \emph{back-edge} from $t_i$ to $s_i$.
To preserve the direction of the arcs in the (undirected) temporal graph, we will subdivide each arc into small paths with ascending time-labels.
As \textsc{Multicut in DAGs} asks for a vertex set, we also need to subdivide each nonterminal vertex $v$ into two new vertices $v_{\text{in}}$ and $v_{\text{out}}$ which are connected by one edge.
Then, the vertex~$v$ is in the cut-set of the original problem if the edge between $v_{\text{in}}$ and $v_{\text{out}}$ is in the solution edge set of the STFES instance.

Before stating our reduction, we introduce two auxiliary concepts.
First, when we want to exclude edges from (S)TFES/(S)TFCS solutions, we will employ a gadget we call \emph{heavy time-edge} which connects two vertices using $k+1$ parallel paths.

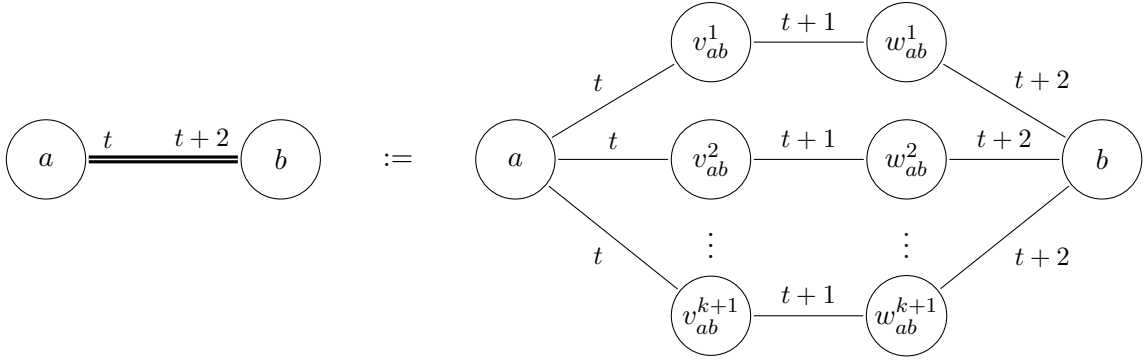
\begin{figure}[t]
  \centering
  \tikzstyle{alter}=[circle, minimum size=30pt, draw, inner sep=1pt]
    \centering
    \begin{tikzpicture}[scale=1.3,auto, >=stealth',shorten <=1pt, shorten >=1pt,scale=.8]
    \tikzstyle{majarr}=[draw=black]
    \node[alter] at (-6,0) (a0) {$a$};
	\node[alter] at (-3,0) (b0) {$b$};
	\node at (-1.5,0) (def) {$:=$};
	\draw[majarr, style=double, style=very thick] (a0) -- (b0) node [near start]{\small{$t\hphantom{+1}$}} node [near end]{\small{$t+2$}};
		      
		\node[alter] at (0,0) (a) {$a$};
		\node[alter] at (7.5,0) (b) {$b$};
		\node[alter] at (2.5,1.5) (v_ab_1) {$v_{ab}^1$};
		\node[alter] at (2.5,0) (v_ab_2) {$v_{ab}^2$};
		\node at (2.5,-1) (v_ab_dots) {$\vdots$};
		\node[alter] at (2.5,-2) (v_ab_k1) {$v_{ab}^{k+1}$};
		
		\node[alter] at (5,1.5) (w_ab_1) {$w_{ab}^1$};
		\node[alter] at (5,0) (w_ab_2) {$w_{ab}^2$};
		\node at (5,-1) (w_ab_dots) {$\vdots$};
		\node[alter] at (5,-2) (w_ab_k1) {$w_{ab}^{k+1}$};
		
		\draw[majarr] (a) -- (v_ab_1) node [midway]{\small{$t$}};
		\draw[majarr] (a) -- (v_ab_2) node [midway]{\small{$t$}};
		\draw[majarr] (a) -- (v_ab_k1) node [midway, below left]{\small{$t$}};
		
		\draw[majarr] (v_ab_1) -- (w_ab_1) node [midway]{\small{$t+1$}};
		\draw[majarr] (v_ab_2) -- (w_ab_2) node [midway]{\small{$t+1$}};
		\draw[majarr] (v_ab_k1) -- (w_ab_k1) node [midway]{\small{$t+1$}};
		
		\draw[majarr] (w_ab_1) -- (b) node [midway]{\small{$t+2$}};
		\draw[majarr] (w_ab_2) -- (b) node [midway]{\small{$t+2$}};
		\draw[majarr] (w_ab_k1) -- (b) node [midway,below right]{\small{$t+2$}};
    \end{tikzpicture}

\caption{Heavy time-edge $\hte (a,b,t)$.}
\label{fig:heavy_time_edge}
\end{figure}

\begin{definition}[Heavy time-edge]
Let $\mathcal{I}=(\TG=(V,\TE,\tau), k)$ be an instance of (S)TFES or (S)TFCS.
For $a,b \in V$ and $t \leq \tau - 2$, a heavy time-edge of $\TG$ is a subgraph $\hte (a,b,t) := (V_h, \TE_h, \tau_h = t+2)$ connecting vertex $a$ to vertex~$b$ with 
\[ V_h = \{a,b\} \cup \{v_{ab}^i, w_{ab}^i \mid i \in [k+1]\} \quad \text{and} \]
\[ \TE_h = \{(\{a,v_{ab}^i\},t), (\{v_{ab}^i, w_{ab}^i\},t+1), (\{w_{ab}^i,b\},t+2) \mid i \in [k+1] \}. \]
 \label{def:heavy_time_edge}
\end{definition}

The construction is shown in \Cref{fig:heavy_time_edge}.
Let $e^h:=\hte (a,b,t)$ be a heavy time-edge.
Due to the time-labels, the gadget only connects $a$ to $b$ (and not $b$ to~$a$) which we will use to model (directed) arcs with (undirected) time-edges.
For (S)TFES/(S)TFCS solutions, it is easy to see that if there is a temporal path from $b$ to $a$, then the $k+1$ cycles going through $e^h$ cannot be disconnected by removing edges inside the gadget.
Thus, without loss of generality we can assume  that a given solution contains no edges from $e^h$.
We also note that, for both temporal path models (i.e., strict and non-strict), it is possible to design a smaller gadget with identical properties, but we opted to use one which works for both models simultaneously.

Second, in order to assign time-labels while preserving all paths of the input graph~$D$, we will use an \emph{acyclic ordering} (also known as topological ordering) of~$D$.
For a directed graph $D=(V,A)$, an acyclic ordering $<$ is a linear ordering of the vertices with the property $(v,w) \in A \Rightarrow v < w$.
In other words, if we place the vertices on a line in the order given by $<$, then all arcs point in one direction (see \Cref{fig:W1_acyc_ord} for an example).
If $D$ is a DAG, then such an ordering exists and can be computed in linear time~\citep[Thm.~4.2.1]{Bang-JensenG02}.
For convenience, we represent this ordering as a function $\pi:V \rightarrow \mathbb{N}$ which maps each vertex to its position in the ordering.
We now have all the ingredients to prove the theorem.

\begin{proof}[Proof of \Cref{thm:STFES_k_W1}]
Let $\mathcal{I}=(D, \mathcal{T}, k)$ be an instance of \textsc{Multicut in DAGs} with terminal vertices $V_\mathcal{T} := \{s_i, t_i \mid (s_i, t_i) \in \mathcal{T}\}$.
We construct an instance $\mathcal{I'}=(\TG,k'=k)$ of (S)TFES as follows.

\begin{figure}[t]%
\centering
  \tikzstyle{alter}=[circle, minimum size=8pt, draw, inner sep=1pt]
    \scalebox{1.2}{
    \begin{tikzpicture}[scale=0.7, auto, >=stealth',shorten <=1pt, shorten >=1pt]
      \tikzstyle{majarr}=[draw=black]
      \tikzstyle{circled}=[circle, minimum size=16pt, draw, inner sep=1pt]
      \node[circled, label={above right:\small{$3$}}] at (0,0) (a) {$a$};
	\node[circled, label={above right:\small{$2$}}] at (2,0) (b) {$b$};
	\node[circled, label={above right:\small{$1$}}] at (4,0) (c) {$c$};
	\node[circled, label={above right:\small{$7$}}] at (0,-2) (d) {$d$};
	\node[circled, label={above right:\small{$4$}}] at (2,-2) (e) {$e$};
	\node[circled, label={above right:\small{$5$}}] at (4,-2) (f) {$f$};
	\node[circled, label={above right:\small{$8$}}] at (0,-4) (g) {$g$};
	\node[circled, label={above right:\small{$9$}}] at (2,-4) (h) {$h$};
	\node[circled, label={above right:\small{$6$}}] at (4,-4) (j) {$j$};
	
	\draw[->] (c)--(b);
	\draw[->] (b)--(a);
	\draw[->] (b)--(e);
	\draw[->] (a)--(e);
	\draw[->] (a)--(d);
	\draw[->] (d)--(g);
	\draw[->] (g)--(h);
	\draw[->] (j)--(h);
	\draw[->] (e)--(h);
	\draw[->] (e)--(f);
	\draw[->] (f)--(j);
	
    \end{tikzpicture}}

\bigskip   

    \scalebox{1.0}{
    \begin{tikzpicture}[auto, >=stealth',shorten <=1pt, shorten >=1pt]
    \tikzstyle{majarr}=[draw=black]
    \tikzstyle{circled}=[circle, minimum size=16pt, draw, inner sep=1pt]
    \node[circled] at (0,0) (c) {$c$};
	\node[circled] at (1,0) (b) {$b$};
	\node[circled] at (2,0) (a) {$a$};
	\node[circled] at (3,0) (e) {$e$};
	\node[circled] at (4,0) (f) {$f$};
	\node[circled] at (5,0) (j) {$j$};
	\node[circled] at (6,0) (d) {$d$};
	\node[circled] at (7,0) (g) {$g$};
	\node[circled] at (8,0) (h) {$h$};
	
	\path[->] (c.north) edge[bend left] (b.north);
	\path[->] (b.north) edge[bend left] (a.north);
	\path[->] (b.south) edge[bend right] (e.south);
	\path[->] (a.north) edge[bend left] (e.north);
	\path[->] (a.north) edge[bend left] (d.north);
	\path[->] (d.north) edge[bend left] (g.north);
	\path[->] (g.north) edge[bend left] (h.north);
	\path[->] (j.north) edge[bend left] (h.north);
	\path[->] (e.south) edge[bend right] (h.south);
	\path[->] (e.north) edge[bend left] (f.north);
	\path[->] (f.north) edge[bend left] (j.north);
	
    \end{tikzpicture}}
\caption{A DAG with values for $\pi (v)$ (derived from an acyclic ordering) for each vertex~$v$ (top) and  with the vertices aligned on a line according to~$\pi$ (bottom).}
\label{fig:W1_acyc_ord}
\end{figure}
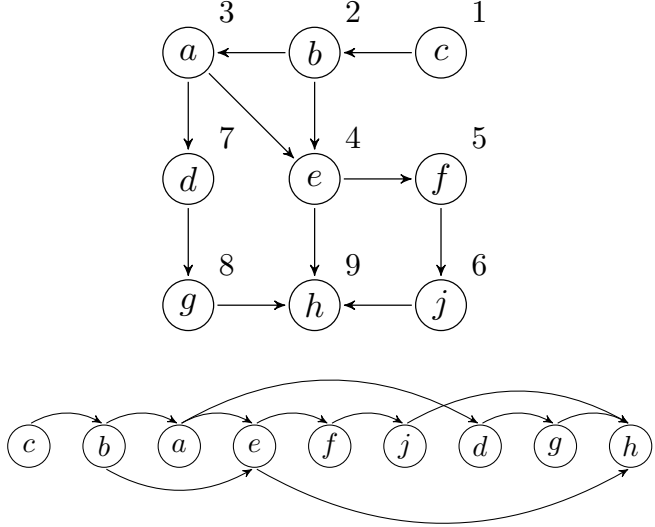

\begin{enumerate}
\item We compute an acyclic ordering of the vertices $V := V(D)$ and store it as a function $\pi:V \rightarrow \mathbb{N}$ (see \Cref{fig:W1_acyc_ord}).
We use $\pi$ to transform $D$ into an equivalent temporal graph $\TG^1=(V',\TE,\tau = 4\abs{V})$ by replacing each arc $(v,w) \in A$ with the heavy time-edge $e^{vw}:=\hte (v,w,4\pi (v)+1)$.
It is easy to verify that, for two vertices $s,t \in V$, the graph $D$ contains a path from vertex $s$ to vertex $t$ if and only if $\TG$ contains a temporal path from $s$ to $t$.
Note that starting each heavy time-edge with time-label $4\pi (v)+1$ leaves layer $4\pi (v)$ empty which we will use in the next step.

\item In $\TG^1$, we replace each nonterminal vertex $v \in V \setminus V_\mathcal{T}$ with two new vertices $v_{\text{in}}$ and $v_{\text{out}}$ connected by time-edge $e^v=(\{v_{\text{in}},v_{\text{out}}\}, 4\pi (v))$ and update the edges adjacent to $v$ as follows.
For each (incoming) edge of the form $e_{uv}:=(\{u,v\}, t)$ with $t < 4\pi (v)$, replace $e_{uv}$ with $(\{u,v_{\text{in}}\}, t)$.
For each (outgoing) edge of the form $e_{vw}:=(\{v,w\}, 4\pi (v)+1)$, replace $e_{vw}$ with $(\{v_{\text{out}},w\}, 4\pi (v)+1)$.
Let $\TG^2$ denote the resulting graph.
Clearly, for two vertices $s,t \in V$, removing $v$ in $\TG^1$ disconnects all $(s,t)$-paths if and only if removing $e^v$ in $\TG^2$ disconnects all $(s,t)$-paths.

\item We obtain $\TG^3=\TG$ by adding a back-edge $\hte (t_i,s_i,\tau_{\text{be}})$ with $\tau_{\text{be}} = 4\abs{V+1}$ for each terminal pair $(s_i, t_i)$.
Since there is a temporal path from $s_i$ to $t_i$, this creates at least one cycle for each terminal pair.

\end{enumerate}

\begin{figure}[t]%
\centering
  \tikzstyle{alter}=[circle, minimum size=8pt, draw, inner sep=1pt]
    \centering
    \scalebox{.78}{
    \begin{tikzpicture}[auto, >=stealth',shorten <=1pt, shorten >=1pt, scale=1]
    \tikzstyle{majarr}=[draw=black]
    \tikzstyle{circled}=[circle, minimum size=1cm, draw, inner sep=1pt]
    \tikzstyle{square}=[minimum size=1cm, draw, inner sep=1pt]
    \node at (-6, 1) (D) {$D$};
    \node[square] at (-8,0) (a_0) {$a$};
    \node[circled] at (-6,0) (b_0) {$b$};
    \node[square] at (-4,0) (c_0) {$c$};
    \draw[->] (a_0) -- (b_0);
    \draw[->] (b_0) -- (c_0);
    \node at (-2,0) (arrow) {$\rightarrow$};
    \node at (3, 1) (G) {$G$};
    \node[square] at (0,0) (a) {$a$};
	\node[circled] at (2,0) (b_in) {$b_\text{in}$};
	\node[circled] at (4,0) (b_out) {$b_\text{out}$};
	\node[square] at (6,0) (c) {$c$};
	
	\draw[majarr, style=double, style=very thick] (a) -- (b_in) node [near start]{\small{$5$}} node [near end]{\small{$7$}};
	\draw[majarr] (b_in) -- (b_out) node[midway] {$8$};
		\draw[majarr, style=double, style=very thick] (b_out) -- (c) node [near start]{\small{$9$}} node [near end]{\small{$11$}};
		\draw[majarr, style=double, style=very thick] (c) -- +(0,-1) -- +(-6,-1) node[near start] {\small{$16$}} node[near end] {\small{$18$}}  -- (a);		
    \end{tikzpicture}}
\caption{Example: Reduction from \textsc{Multicut in DAGs} with input digraph~$D$ and one terminal pair $(a,c)$ to TFES. Double lines represent heavy time-edges (\Cref{def:heavy_time_edge}).}
\label{fig:W1_example}
\end{figure}
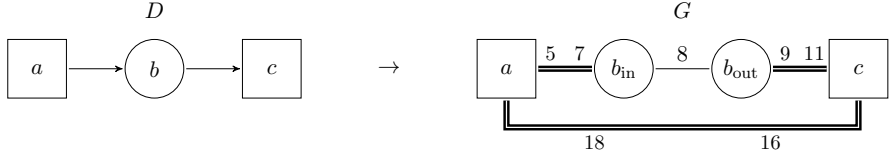

\Cref{fig:W1_example} shows a small example. It is easy to see that the construction can be done in polynomial time.
Now we show that $\mathcal{I}=(D, \mathcal{T}, k)$ is a yes-instance of \textsc{Multicut in DAGs} if and only if $\mathcal{I'}=(\TG,k'=k)$ is a  yes-instance of (S)TFES. 

$(\Rightarrow):$ Let $Z$ be a solution of $\mathcal{I}$.
We claim that $ \TE' = \{ (\{v_{\text{in}},v_{\text{out}}\}, 4\pi (v)) \mid v \in Z \}$ is a solution of $\mathcal{I'}$.
We first show that, for any terminal pair $(s_i,t_i)$, the graph $\TG - \TE'$ contains no temporal path from $s_i$ to $t_i$.
For $\TG^2$, this is easily verified as $D-Z$ contains no $(s_i,t_i)$-path.
In $\TG=\TG^3$, this claim holds if no temporal path from $s_i$ to $t_i$ contains a back-edge $e^{t_j s_j}=\hte (t_j,s_j,\tau_{\text{be}})$ added in step 3.
Due to the starting time-label of the back-edges, no temporal path can contain more than one back-edge, and
if it does contain one back-edge, then this back-edge must be at its end.
Clearly, the temporal path cannot end at both $t_i$ and $s_j$ unless $t_i = s_j$, which we excluded by applying \Cref{lem:W1_KratschPPW}.
Now, assume towards contradiction that $\TG - \TE'$ contains a cycle~$C$.
Since $\TG^2$ was cycle-free, $C$ must use some back edge $e^{t_i s_i}$ introduced in step 3 and, as reasoned above, this back edge must be the last edge of $C$.
However, there is no temporal path from $s_i$ to $t_i$ in $\TG - \TE'$ and, thus, $C$ cannot be a temporal cycle.

$(\Leftarrow):$ Let $\TE'$ be a solution of $\mathcal{I'}$, i.e., $\TG - \TE'$ contains no cycles.
Recall that $V_\mathcal{T} := \{s_i, t_i \mid (s_i, t_i) \in \mathcal{T}\}$ is the set of terminal vertices of $\mathcal{I}$.
As observed above, we may assume that $\TE'$ does not contain any edges from heavy time-edges.
Thus we have $\TE' \subseteq \{e^v \mid v \in V \setminus V_\mathcal{T}\}$ and define the solution for $\mathcal{I}$ as $Z:=\{v \mid e^v \in \TE'\}$.
Assume towards a contradiction that $D-Z$ contains an $(s_i,t_i)$-path for some terminal pair $(s_i,t_i)$.
This path induces a temporal path from $s_i$ to $t_i$ in $\TG-\TE'$
which we can extend back to $s_i$ by appending the back edge $e^{t_i s_i}$ to obtain a cycle in $\TG - \TE'$ and, thus, a contradiction.

For both directions, we have $\abs{Z}=\abs{\TE'} \leq k=k'$ meeting the requirements for the solution size.

As the constructed temporal graph $\TG$ contains no pair of vertices connected by more than one time-edge, we can easily transform a minimal feedback edge set of $\TG$ into a minimal feedback connection set.
Thus, the arguments presented in this proof also hold for (S)TFCS.
\end{proof}

\section{Fixed-Parameter Tractability Results}
\label{sec:algorithms}
After having shown computational hardness for the single parameters solution size~$k$ and lifetime~$\tau$ in \Cref{sec:hardness}, we now consider larger and combined parameters, and present fixed-parameter tractability results. %

\subsection{Parameterization by Number of Vertices}
\label{sec:fptn}
As shown in \Cref{obs:tfcs_v_fpt}, (S)TFCS is trivially fixed-parameter tractable with respect to the number of vertices~$\abs{V}$.
For (S)TFES, however, the same result is much more difficult to show as the size of the search space is only upper-bounded by $2^{\tau (\abs{V}^2-\abs{V})}$.
Here, the dependence on $\tau$ prevents us from using the (brute-force) approach that worked for (S)TFCS.

\begin{theorem}
\textsc{STFES} can be solved in $\bigO(2^{\abs{V}^3}\abs{V}^4 \tau)$ time
and \textsc{TFES} in $\bigO(2^{\abs{V}^3 + \abs{V}^2}\abs{V}^7\tau)$ time,
both requiring $\bigO(2^{\abs{V}^3})$ space.
\label{thm:FPTV}
\end{theorem}

We prove \Cref{thm:FPTV} using a dynamic program which computes the minimum number of time-edges which have to be removed to achieve a specified connectivity at a specified point in time.
The rough idea is that we can efficiently determine which edges need to be removed from layer~$t$
if we know which vertices can reach which other vertices until time $t-1$ and time $t$, respectively.
We would then like to exclude temporal cycles by simply requiring that every vertex is unreachable from itself until time $\tau$.
However, the situation is slightly more complicated as a (even non-trivial) temporal walk from a vertex to itself might simply ``backtrack'' along itself,
and thus not constitute a temporal cycle.
We will handle this complication by introducing \emph{side-trip-free} temporal walks, which are essentially forbidden from backtracking along previously used edges.

Let~$P$ be a walk (or temporal walk) with vertex sequence $v_1, \dots, v_{\ell+1}$.
We call~$P$ \emph{side trip free} if $v_i \neq v_{i+2}$ for all $1 \leq i \leq \ell-1$.
We further denote the second vertex of~$P$ by $\sigma(P) := v_2$ and the penultimate vertex by $\rho(P) := v_\ell$.
In the special case of $P$~being the trivial walk, we define $\rho(P) := v_1$.

The motivation behind side-trip-free temporal walks is due to the following observation.
\begin{observation}
	\label{thm:side-trip-free-is-cycle}
	A temporal graph contains a (strict) temporal cycle
	if and only if it contains a (strict) non-trivial side-trip-free temporal walk from a vertex to itself.
\end{observation}
\begin{proof}
	One implication is immediate.
	For the other direction, assume $P$~to be a (strict) non-trivial side-trip-free temporal walk from a vertex~$v$ to itself.
	Let $v_1, \dots, v_{\ell+1}$ be the vertex sequence of~$P$,
	where $v_1 = v = v_{\ell+1}$.
	Let $i$ be the maximal index such that there is $j > i$ with $v_i = v_j$.
	By also picking~$j$ minimally, we may assume that $v_i, v_{i+1}, \dots, v_{j-1}$ are pairwise distinct.
	Thus, the corresponding subwalk of~$P$ is a temporal cycle.
\end{proof}

The first index of our dynamic programming table will be a
\emph{connectivity table}~$A \in \{0,1\}^{\abs{V} \times \abs{V} \times \abs{V}}$
which is itself indexed by triples of vertices.
We say that a temporal graph~$\TG$ \emph{adheres} to~$A$
if $A_{uvw} = 0$ implies that $\TG$~contains no side-trip-free temporal walk~$P$ from~$u$ to~$w$ with $\rho(P)=v$.
We say that $\TG$ \emph{fully adheres} to~$A$ if this implication is bidirectional,
i.e., if $A_{uvw} = 1$ also implies that there is such a side-trip-free temporal walk.

Next, we define two functions, $\s_absvupdate (G, B, A)$ (\textbf{s}trict \textbf{r}equired \textbf{d}eletions)
and $\absvupdate (G, B, A)$ (\textbf{n}on-strict \textbf{r}equired \textbf{d}eletions),
which return the solution to the following subproblem:
Given connectivity tables~$B$ (\textbf{b}efore) and~$A$ (\textbf{a}fter) and a graph~$G$,
what is the minimum number of edge deletions in~$G$ to ensure that
for any temporal graph~$\TG$ with vertex set~$V(G)$ that fully adheres to~$B$,
the temporal graph obtained by appending the layer~$G$ to~$\TG$ adheres to~$A$?
(Of course, the difference between $\s_absvupdate$ and $\absvupdate$ is whether connectivity is
evaluated for strict or non-strict temporal walks.)
\Cref{fig:FPT_V} illustrates this problem for the strict case.

\begin{figure}[t] 
\centering
  \tikzstyle{alter}=[circle, minimum size=8pt, draw, inner sep=1pt]
    \centering
    \begin{tikzpicture}[auto, >=stealth',shorten <=1pt, shorten >=1pt]
      \tikzstyle{majarr}=[draw=black]
    \coordinate (y) at (0,5);
    \coordinate (x) at (7,0);
    \draw[<->] (y) node[above] {vertices} -- (0,0) --  (x) node[right]
    {time};
    \draw (0.5, 0) node[below] {$1$};

    \node[alter, fill=white, label = left:{$u$}] at (0,4.5) (u0) {};
    \node[alter, fill=white, label = left:{$v$}] at (0,2) (v0) {};
    \node[alter, fill=white, label = left:{$x$}] at (0,3.5) (x0) {};
    \node[alter, fill=white, label = left:{$w$}] at (0,1) (w0) {};

    \node[alter, fill=white, label = right:{$v$}] at (5.5,2) (vt) {};
    \node[alter, fill=white, label = right:{$x$}] at (4.5,3.5) (xt) {};
    \node[alter, fill=white, label = right:{$w$}] at (6.5,1) (wt) {};

    \draw[dotted,thick] (u0) -- +(0.5,0) -- +(1, -1.5) -- +(1.5, -1.5) -- +(2, -3) -- +(2.5, -3) -- +(3, -1) -- +(3.5, -1) -- +(4, -0.5) -- (xt);
    \node at (2.3,1) (b_label) {``$B_{uxv}=1$''};
    
    \draw (5, 0) node[below] {$t-1$};
    \draw (6, 0) node[below] {$t$};
    \draw[majarr] (vt) to (wt) ;
    \draw[majarr] (xt) to (vt);
    \end{tikzpicture}
    \caption{Illustration of the subproblem solved by $\s_absvupdate(G, B, A)$.
		If we have $A_{uvw} = 0$ (i.e., we want no side-trip-free temporal path~$P$ from $u$ to $w$ with $\rho(P) = v$ to exist at time $t$),
		and if we further have $B_{uxv} = 1$ (i.e., there is a side-trip-free temporal path $P'$ from $u$ to $v$ with $\rho(P') = x \neq w$ as indicated by the dotted line),
		then the time-edge $(\{v, w\},t)$ has to be removed from the graph.
	}
\label{fig:FPT_V}
\end{figure}
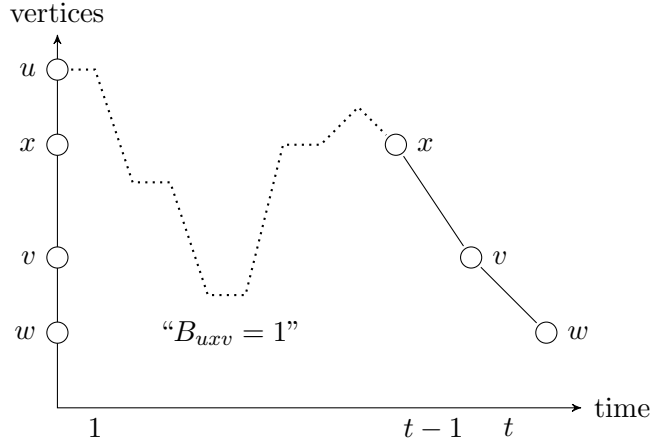

\begin{lemma}
\label{def:s_absvupdate}
Let $G=(V,E)$ be a static graph with $\abs{V}=n$ and let $A,B \in \{0,1\}^{n \times n \times n}$ be two connectivity tables.
Then
\begin{align*}
	\s_absvupdate (G, B, A) = 
	\begin{cases}
		\infty & \text{if } \exists u,v,w: B_{uvw} > A_{uvw} \\
		\abs*{\{ \{v, w\} \in E \mid \exists u: A_{uvw} = 0 \land  \exists x \neq w: B_{uxv} = 1\}} & \text{otherwise}
	\end{cases}
\end{align*}
\end{lemma}

\begin{proof}
Let $\TG$ be any temporal graph fully adhering to connectivity~$B$
and $\TG'$ obtained from appending~$G$ to~$\TG$ (say as layer~$\tau$).

Clearly, deleting edges from~$G = \TG'_\tau$ cannot destroy any temporal walks that exist in $\TG$.
Thus if $B_{uvw} = 1$ but $A_{uvw} = 0$ for some $u,v,w$, then no number of deletions suffices.
So suppose now $B_{uvw} \leq A_{uvw}$ for all $u, v, w$.

Let $\TG''$ be obtained from $\TG'$ by deleting from $\TG'_\tau$ the edges in
\[
	E' := \{ \{v, w\} \in E \mid \exists u: A_{uvw} = 0 \land  \exists x \neq w: B_{uxv} = 1\}.
\]
To show that $\TG''$ adheres to~$A$, it is sufficient to show for each entry~$A_{uvw}=0$
that $\TG''$ contains no side-trip-free strict temporal walk~$P$ from $u$ to $w$ with $\rho(P) = v$ which arrives exactly at time~$\tau$
(because $A_{uvw}=0$ implies $B_{uvw}=0$).
If such a walk~$P$ exists, then it must contain the edge $\{v, w\}$ at time~$\tau$
and further it must contain a side-trip-free strict temporal walk $P' \subseteq P$ from $u$ to~$v$
with $\rho(P') \neq w$.
(If we had $\rho(P') = w$, then $P$ would not be side trip free.)
But then, $\{v, w\} \in E'$ by definition of $E'$. 
This proves $\s_absvupdate(G,B,A) \leq \abs{E'}$.

To see that $\s_absvupdate(G,B,A) \geq \abs{E'}$,
let $\tilde{E}$ be any subset of~$E(G)$ and $\tilde{G}$~be obtained from~$\TG'$ by deleting the edges in~$\tilde{E}$ from layer~$\tau$.
If $\tilde{G}$ adheres to~$A$,
then we claim that $\tilde{E} \supseteq E'$.
Suppose not, then there are $u, v, w,x$ with $x \neq w$, $A_{uvw} = 0$, $B_{uxv} = 1$, and $\{v, w\} \in E(G) \setminus \tilde{E}$.
Since $\TG$ fully adheres to~$B$, there is in~$\TG$ a side-trip-free strict temporal walk~$Q$ 
from $u$ to~$v$ with $\rho(v) = x$.
Appending the time-edge~$(\{v, w\},\tau)$ to~$Q$ then shows that $\tilde{\TG}$ does not adhere to~$B$.
\end{proof}

\Cref{def:s_absvupdate} shows that $\s_absvupdate$ can be computed in $\bigO(\abs{V}^4)$~time
by iterating over all 4-vertex tuples~$(u,v,w,x)$.

In the non-strict case, a temporal path can successively use multiple edges from~$G$.
Thus, it is not possible to consider each entry $A_{uvw}=0$ separately (a single edge might be part of multiple unwanted temporal walks).
Instead, we have to find an optimal edge-cut disconnecting all ``problematic'' pairs in $G$.

\begin{lemma}
\label{def:absvupdate}
Let $G=(V,E)$ be a static graph with $\abs{V}=n$ and let $A,B \in \{0,1\}^{n \times n}$ be two connectivity tables.

If there exist $u,v,w$ such that $B_{uvw} > A_{uvw}$, then $\absvupdate (G, B, A)=\infty$.
Otherwise, $\absvupdate(G, B, A) = \abs{E'}$
where $E' \subseteq E(G)$ is a minimum-size set which intersects
for all $u,v,w,s,x$ with $A_{uvw} = 0$ and $B_{uxs} = 1$,
every non-trivial side-trip-free $s$-$w$ walk~$P$ with $\rho(P) = v$ and $\sigma(P) \neq x$.
\end{lemma}
\begin{proof}
Let $\TG$ be any temporal graph fully adhering to connectivity~$B$
and $\TG'$ obtained from appending~$G$ to~$\TG$ (say as layer~$\tau$).

As in the proof of \cref{def:s_absvupdate}, if there are~$u, v, w$ with $B_{uvw} = 1$ and $A_{uvw} = 0$, then (and only then) no amount of edge deletions suffices.
So suppose now otherwise.

Let $E' \subseteq E(G)$ be as stated above
and $\TG''$ be obtained from~$\TG'$ by deleting the edges in~$E'$ from~$\TG'_\tau$.
Suppose for contradiction that $\TG''$ does not adhere to~$A$,
i.e., there are~$u, v, w$ with $A_{uvw} = 0$ but a side-trip-free temporal walk~$P$ from~$u$ to~$w$ with $\rho(P)=v$ exists in~$\TG''$.
Since~$B_{uvw} = 0$, $P$~must arrive at time~$\tau$.
Let~$s$ be the last vertex $P$~reaches before time~$\tau$
and let $P', P''$ be the temporal sub-walks ending resp.\ starting at~$s$.
Set $x := \rho(P')$, then clearly $B_{uxs} = 1$
and $\sigma(P'') \neq x$.
Furthermore $\rho(P'') = \rho(P) = v$.
Together, this implies that $P''$ intersects $E'$ by definition of~$E'$.
This shows $\absvupdate(G, B, A) \leq \abs{E'}$.

To see that $\absvupdate(G, B, A) \geq \abs{E'}$,
let $\tilde{E}$ be any subset of~$E(G)$ and $\tilde{\TG}$~be obtained from~$\TG'$ by deleting the edges in~$\tilde{E}$ from layer~$\tau$.
If $\tilde{\TG}$ adheres to~$A$, 
then we claim that $\tilde{E}$ intersects
for all $u,v,w,s,x$ with $A_{uvw} = 0$ and $B_{uxs} = 1$,
every non-trivial side-trip-free $s$-$w$ walk~$P$ with $\rho(P) = v$ and $\sigma(P) \neq x$.
(Of course, this implies $\abs{\tilde{E}} \geq \abs{E'}$.)
So suppose that this is not the case.
Since $\TG$~fully adheres to~$B$, there is in~$\TG$ a side-trip-free temporal walk~$Q$ from $u$ to~$s$ with $\rho(Q) = x$.
Concatenating $Q$ and~$P$ then produces a side-trip-free temporal walk~$Q'$ from $u$ to~$w$ with $\rho(Q') = \rho(P) = v$
and $A_{uvw} = 0$,
which contradicts $\tilde{\TG}$ adhering to~$A$.
\end{proof}

\Cref{def:absvupdate} gives us a way to compute $\absvupdate(G,B,A)$:

\begin{lemma}
Function~$\absvupdate (G,B,A)$ can be computed in $\mathcal{O}(2^{\abs{V}^2} \cdot \abs{V}^7)$ time.
\label{lem:absv_rd_time}
\end{lemma}
\begin{proof}
The number of subsets of $E$ is at most $2^{\abs{V}^2-\abs{V}}$.
For each subset~$E' \subseteq E$ and every tuple~$(u,v,w,s,x)$,
we can check 
in $\bigO(\abs{E(G)}) \subseteq \bigO(\abs{V}^2)$~time
whether $G$ contains a non-trivial side-trip-free $s$-$w$ walk~$P$ with $\rho(P) = v$ and $\sigma(P) \neq x$.
We do this by checking whether $G$ contains the edge $\{v, w\}$ as well as an $s$-$v$ walk which does not use~$w$.
Thus, $\absvupdate (G, B, A)$ can be computed in $\mathcal{O}(2^{\abs{V}^2} \cdot \abs{V}^7)$ time.
\end{proof}

We can now define the dynamic program which we will use to prove \Cref{thm:FPTV}.
Let $n := \abs{V}$ and $A \in \{0,1\}^{n \times n \times n}$ be a connectivity table.
Recall that $\TG_{[t]}$ refers to the first $t$~layers of~$\TG$.
We define the table entry $\T(A,t) \in \mathbb{N}$ as the minimum number of time-edges which have to be removed from $\TG_{[t]}$
in order for the resulting temporal graph to adhere to~$A$.

Then, $\T$~is as follows. 

\newcounter{subeq}
\renewcommand{\thesubeq}{\theequation\alph{subeq}}
\newcommand{\newsubeqblock}{\setcounter{subeq}{0}\refstepcounter{equation}}
\newcommand{\mysubeq}{\refstepcounter{subeq}\tag{\thesubeq}}

\begin{lemma}
\label{thm:recursion-formula}
$\T$ as defined above satisfies the following recursive formula.
\begin{align}
\T(A,0) &= 0\\
\newsubeqblock
\mysubeq \text{strict paths: }\T(A,t) &= \min_{B \in \{0,1\}^{n \times n \times n}} \T(B,t-1) + \s_absvupdate (G_t, B, A) \label{eq:s_absvupdate} & \forall t > 0\\
\mysubeq \text{non-strict paths: }T(A,t) &= \min_{B \in \{0,1\}^{n \times n \times n}} \T(B,t-1) + \absvupdate (G_t, B, A) & \forall t > 0
\label{eq:absvupdate}
\end{align}
\label{lem:absv_T_correct}
\end{lemma}

\begin{proof}
We prove the lemma for the strict case via induction over $t$. The non-strict case works analogously.
The correctness of the initialization $T(A,0)=0$ is easy to see
since the layerless temporal graph $\TG_{[0]}$ contains no non-trivial temporal paths.

Let now $t > 0$.
It is easy to see that 
\[
	\T(A, t) \leq \min_{B \in \{0,1\}^{n \times n \times n}} \T(B,t-1) + \s_absvupdate (G_t, B, A)
\]
because, for any choice of $B$,
$\T(B,t-1) + \s_absvupdate (G_t, B, A)$
is the minimum number of edge deletions required to have $\TG_{[t-1]}$ fully adhere to~$B$
and have $\TG_{[t]}$ adhere to~$A$.

To prove the reverse inequality, let $\TE'$ be a minimum-size set of time-edges
whose removal from~$\TG_{[t]}$ ensures that the resulting temporal graph $\TG'$ adheres to connectivity table~$A$.
Partition~$\TE'$ into $\tilde{\TE} \cup \hat{\TE}$ with $\hat{\TE}$~containing the time-edges at time~$t$.
Let~$\tilde{B}$ be the connectivity table which $\TG'_{[t-1]}$ fully adheres to.

We claim that $\abs{\tilde{\TE}} = \T(\tilde{B}, t-1)$.
Clearly $\abs{\tilde{\TE}} \geq \T(\tilde{B}, t-1)$.
Furthermore, if $\abs{\tilde{\TE}} > \T(\tilde{B}, t-1)$ were true,
then one could take a set $X$ of $\T(\tilde{B}, t-1)$ time-edges whose deletion
causes $\TG_{[t-1]}$ to adhere to~$\tilde{B}$.
The set $X \cup \hat{\TE}$ would then contradict the minimality of $\TE'$.

Note that $\s_absvupdate(G_t, \tilde{B}, A) = \abs{\hat{\TE}}$ by definition of $\s_absvupdate$.
Therefore, 
\[
	\T(A, t) =  \T(\tilde{B},t-1) + \s_absvupdate (G_t, \tilde{B}, A) \geq  \min_{B \in \{0,1\}^{n \times n \times n}} \T(B,t-1) + \s_absvupdate (G_t, B, A).\qedhere
\]
\end{proof}

We now have all required ingredients to prove \Cref{thm:FPTV}. 

\begin{proof}[Proof of \Cref{thm:FPTV}]

Let $(\TG,k)$ be an instance of (S)TFES and $n$~the number of vertices.
Further, let~$A^* \in \{0, 1\}^{n \times n \times n}$ be the connectivity table with~$A^*_{uvw}=0$ if and only if $u = w \neq v$.
Then it follows from the definition of $\T$ that $\T(A^*,\tau) \leq k$
if and only if there is a set~$\TE'$ of at most~$k$ time-edges
such that $\TG - \TE'$ contains no (strict) non-trivial side-trip-free temporal walk from a vertex to itself.
By \cref{thm:side-trip-free-is-cycle} this is equivalent to $\TG-\TE'$ being free of (strict) temporal cycles.
Thus $\T(A^*,\tau) \leq k$ if and only if $(\TG,k)$ is a yes-instance.

We compute $\T(A^*, \tau)$ by means of \cref{lem:absv_T_correct}.
This requires to compute $2^{n^3}\tau$ table entries,
each taking $\bigO(n^4)$~time in the strict case (\cref{def:s_absvupdate})
and $\bigO(2^{n^2} n^7)$~time in the non-strict case (\cref{def:absvupdate}).

As it suffices to only keep the table entries for $t-1$ and $t$ in memory at any given time,
the computation requires $\bigO(2^{n^3})$~space.
\end{proof}

We note that our dynamic program indeed solves the optimization variant of (S)TFES.
That is, given a temporal graph $\TG$, the dynamic program finds the smallest~$k$ for which $(\TG,k)$~is a yes-instance of the decision variant stated defined in \Cref{sec:intro}.
As shown in the previous proof, we can easily use this result to solve any instance $(\TG,k')$ of the decision variant by comparing~$k'$ to~$k$.

For the ease of presentation, we did not store the actual solution, that is, the feedback edge set of size~$T(A,t)$.
However, the functions $\s_absvupdate (G_t, B, A)$ and $\absvupdate (G_t, B, A)$ can easily be changed to return the solution edge sets for each layer $t$.
Note that this is possible without changing the asymptotic running time
by storing, in each table entry $T(A, t)$, a reference to the entry $T(B, t-1)$ for which the minimum of the recursive formula of \cref{thm:recursion-formula} is assumed.
In this way, it can be avoided to copy the corresponding solution sets over and over.

\subsection{Parameterization by Treewidth and Lifetime}
\newcommand{\inc}{\operatorname{inc}}
\newcommand{\pres}{\operatorname{pres}}
\newcommand{\tim}{\operatorname{time}}
\newcommand{\edge}{\operatorname{edge}}
In this last part, we show that all our problem variants are fixed-parameter tractable when parameterized by the combination of the treewidth of the underlying graph and the lifetime.
To this end we employ an optimization variant of Courcelle's famous theorem
on graph properties expressible in monadic second-order (MSO) logic \cite{ARNBORG1991308,courcelle2012graph} and apply it in the temporal setting~\cite{Flu+20}.

\begin{theorem}
\label{thm:mso}
\textsc{(S)TFES} and \textsc{(S)TFCS} are fixed-parameter tractable when parameterized by the combination of the treewidth of the underlying graph and the lifetime.
\end{theorem}

To prove this result we require an auxiliary (static) graph $\Ss$ whose vertex set is the disjoint union of 
\begin{itemize}
	\item the set~$V$ of vertices of $\TG$,
	\item the set~$E := E(G_\downarrow)$ of underlying edges,
	\item the set~$[\tau]$ of points in time, and
	\item the set~$\TE$ of time-edges.
\end{itemize}
Its edges are given by the (disjoint union of) the following binary relations,
where we write $R(e,v)$ as a shortcut for $(e, v) \in R$:
\begin{itemize}
	\item the incidence relation $\inc \subseteq E \times V$ where $\inc(e, v) \iff v \in e$,
	\item the time relation $\tim \subseteq \TE \times [\tau]$ where $\tim((e, t), t') \iff t = t'$,
	\item the edge relation $\edge \subseteq \TE \times E$ where $\edge((e, t), e') \iff e = e'$, and
	\item the presence relation $\pres \subseteq E \times [\tau]$ where $\pres(e, t) \iff (e, t) \in \TE$.
\end{itemize}

A \emph{monadic second-order (MSO) formula} over $\Ss$ 
is a formula that uses
\begin{itemize}
\item the above relations,
\item the logical operators $\land$, $\lor$, $\neg$, $=$, and parentheses,
\item a finite set of variables, each of which is either taken as an element or a subset of $V(\Ss)$, and
\item the quantifiers $\forall$ and $\exists$.
\end{itemize}

Additionally we will use some folklore shortcuts such as $\not=$, $\subseteq$, $\in$, and~$\setminus$, which can themselves be replaced by MSO formulas.

By the following theorem, for any property that can be expressed by an MSO formula,
a minimum subset that satisfies it can be computed in linear time.

\begin{theorem}[{\citet[Thm.~5.6]{ARNBORG1991308}}]
	\label{thm:mso-opt}
	There exists an algorithm that, given
	\begin{itemize}
		\item an MSO formula $\phi$ with free variables $X_1,\dots,X_r$,
		\item an affine function $\alpha(x_1,\dots, x_r)$, and
		\item a graph $G$ together with a tree decomposition of width $w$,
	\end{itemize}
	finds the minimum of~$\alpha(|X_1|,\dots,|X_r|)$ over all~$X_1,\dots,X_r \subseteq V(G)$ for which formula~$\phi$ is satisfied on $G$.
	The running time is~$f(|\phi|,w) \cdot |G|$, where~$|\phi|$ is the length of~$\phi$ and $f$~some computable function.
\end{theorem}

\begin{proof}[Proof of \Cref{thm:mso}]
Let $(\TG =(V,\TE,\tau),k)$ be a problem instance of one of our problem variants.
We will construct an MSO formula over the auxiliary graph $\Ss$ defined above
that verifies whether any given set is a solution to this instance.

First, we observe that the treewidth of $\Ss$ is bounded in terms of
$\tw(G_\downarrow) + \tau$.
To this end let $(T, \{B_t\}_{t \in V(T)})$ be an optimal tree decomposition of $G_\downarrow$ where $T$ is a tree and $B_t \subseteq V(G)$ for $t \in V(T)$.
Then $(T, \{B'_t\}_{t \in V(T)})$ with $B'_t := B_t \cup [\tau] \cup (E \cap \binom{B_t}{2}) \cup \{(e, t) \in \TE \mid e \in E \cap \binom{B_t}{2}\}$
is a tree decomposition of $\Ss$ of width at most $\bigO(\tau \cdot
\tw(G_\downarrow)^2)$.
Note that a suitable tree decomposition of $\Ss$ can be computed in $f(\tau +
\tw(G_\downarrow)) \cdot \abs{\Ss}$~time for some function $f$ \cite{Bodlaender96}.

Second, we construct the MSO formulas to express our problem variants.
We do this by assembling them from several auxiliary subformulas encoding simple properties.
We leave it to the reader to verify that each of these formulas agrees with their description.
\newcommand{\tadj}{\operatorname{tadj}}
\newcommand{\cycle}{\operatorname{cycle}}
\newcommand{\conn}{\operatorname{conn}}
\newcommand{\conngraph}{\operatorname{conngraph}}
\newcommand{\tconn}{\operatorname{tconn}}
\newcommand{\ttconn}{\operatorname{ttconn}}
\newcommand{\teelem}{\operatorname{teelem}}
\begin{itemize}
	\item $\tadj(v, w, t)$~tests whether two vertices~$v, w \in V$ are adjacent at time~$t$:
		\[ \tadj(v, w, t) := \exists e \in E : \inc(e,v) \land \inc(e,w) \land \pres(e,t) \]
	\item $\conngraph(X, E')$~tests whether the subgraph $(X, E' \cap X^2)$ of~$G_\downarrow$ is connected:
		\[ \conngraph(X, E') := \forall \emptyset \neq Y \subset X \exists x\in X \setminus Y \exists y \in Y \exists e \in E' : \inc(e,x) \land \inc(e,y) \]
	\item $\conn(v, w, E')$ tests whether the two vertices~$v, w \in V$ are connected by a path that only uses edges from~$E'\subseteq E$:
		\[\conn(v,w,E'):=\exists X\subseteq V: \conngraph(X, E') \land  v\in X\land w\in X\]
	\item $\ttconn(v, w, t, \TE')$ tests whether two vertices~$v, w \in V$ are connected by a path that uses only edges from $\{e \mid (e,t) \in \TE'\}$ for some given $t$ and $\TE' \subseteq \TE$:
		\[ \ttconn(v,w,t,\TE') := \exists E'\subseteq E \forall e\in E'\exists \eps \in \TE': \edge(\eps,e) \land \tim(\eps,t) \land \conn(v,w,E') \]
	\item $\tconn(v, w, t, E')$ tests whether two vertices~$v, w \in V$ are connected by a path that uses only edges from $E' \cap E_t$:
		\[ \tconn(v,w,t,E') := \exists \TE' \subseteq \TE \forall \eps \in \TE' \exists e \in E': \edge(\eps, e) \land \ttconn(v,w,t,\TE')  \]
	\item $\teelem(v, w, t, \TE')$ tests whether $(\{v, w\}, t) \in \TE'$ for some given $\TE' \subseteq \TE$:
		\[ \teelem(v, w, t, \TE') := \exists \eps \in \TE' \exists e \in E : \edge(\eps, e) \land \inc(v, e) \land \inc(w, e) \land \tim(\eps, t) \]	
\end{itemize}
Using these subformulas of constant size,
we can now construct formulas expressing the existence of a temporal cycle for each of our four problem variants as follows.
\begin{itemize}
	\item $\cycle_{\text{SC}}(E')$ tests whether there is a strict temporal cycle using only time-edges whose underlying edges are contained in $E' \subseteq E$: 
		\begin{align*}
			\cycle_{\text{SC}}(E') &:= \exists v_1, v_2, \dots, v_\tau \in V:
			\bigvee_{t^*=1}^{\tau-2} \left(\tadj(v_\tau, v_{t^*},  t^*) 
			 \land \bigwedge_{t=t^*+1}^{\tau}  
			\left( v_t = v_{t-1} \lor (*) \right)\right) \,,
			\intertext{where the subformula}
			(*) &:= \tadj(v_t, v_{t-1}, t) \land \{v_t, v_{t-1}\} \in E' \land \{v_t, v_{t-1}\} \neq \{v_{t^*}, v_\tau\}
		\end{align*}
		tests whether $v_t$ and $v_{t-1}$ are connected at time~$t$ by a time-edge whose underlying edge is contained in $E' \setminus \{\{v_{t^*}, v_\tau\}\}$.
		
		Here a satisfying sequence of vertices $v_\tau, v_{t^*}, v_{t^*+1}, \dots, v_\tau$ (sans repetitions) forms a closed temporal path.
		Note that this sequence might not be a temporal cycle,
		but it will contain a subsequence which does
		because the underlying edge $\{v_{t^*}, v_\tau\}$ is used exactly once.

	\item $\cycle_{\text{C}}(E')$ tests whether there is a non-strict temporal cycle using only time-edges whose underlying edges are contained in $E' \subseteq E$:
		\begin{align*}
			\cycle_{\text{C}}(E') &:= \cycle_{\text{SC}}(E'), \text{ but with $(*)$ replaced by $(**)$}\\
			({*}{*}) &:= \tconn(v_t, v_{t-1}, t, E' \setminus \{\{v_{t^*}, v_\tau\}\})
		\end{align*}
		The subformula~$(**)$ tests whether $v_t$ and $v_{t-1}$ are connected at time~$t$ by a path whose underlying edges are all contained in $E' \setminus \{\{v_{t^*}, v_\tau\}\}$.

	\item $\cycle_{\text{SE}}(\TE')$ tests whether there is a strict temporal cycle using only time-edges in $\TE' \subseteq \TE$:
		\begin{align*}
			\cycle_{\text{SE}}(\TE') &:= \cycle_{\text{SC}}(\TE'), \text{ but with $(*)$ replaced by $(*{*}*)$} \\
			({*}{*}{*}) &:= \teelem(v_t, v_{t-1}, t, \TE') \land \{v_t, v_{t-1}\} \neq \{v_{t^*}, v_\tau\}
		\end{align*}
		The subformula~$(*{*}*)$ tests whether $v_t$ and $v_{t-1}$ are connected at time~$t$ by a time-edge from $\TE' \setminus \{(\{v_{t^*}, v_\tau\}, t)\}$.
	
	\item $\cycle_{\text{E}}(\TE')$ tests whether there is a non-strict temporal cycle using only time-edges in $\TE' \subseteq \TE$:
		\begin{align*}
			\cycle_{\text{E}}(\TE') &:= \cycle_{\text{SC}}(\TE'), \text{ but with $(*)$ replaced by $({**}{**})$} \\
			({*}{*}{*}{*}) &:= \ttconn(v_t, v_{t-1}, t, \TE' \setminus \{(\{v_{t^*}, v_\tau\}, t\})
		\end{align*}
		The subformula~$({**}{**})$ tests whether $v_t$ and $v_{t-1}$ are connected at time~$t$ by a sequence of time-edges from $\TE' \setminus \{(\{v_{t^*}, v_\tau\}, t\}$.
\end{itemize}
It is easy to check that the sizes of the formulas are in $\bigO(\tau^2)$. 
Based on these formulas, we can now give formulas that check whether $\TE'$~is a (strict) temporal feedback edge set, 
respectively whether $E'$~is a (strict) temporal feedback connection set:
\begin{align*}
\phi_{\text{(S)TFES}}(\TE') &:=  \neg
\cycle_{\text{(S)E}}(\TE\setminus\TE')
\\ \phi_{\text{(S)TFCS}}(E') &:= \neg \cycle_{\text{(S)C}}(E\setminus
E')
\end{align*}
The result now follows from \Cref{thm:mso-opt} (for $\alpha(x)=x$)
since $\abs{\Ss} \in \bigO(\tau + \abs{\TG})$, $\tw(\Ss) \in
\bigO(\tau \cdot \tw(G_\downarrow)^2)$, and
$\abs{\phi_{\text{(S)TFES}}},\abs{\phi_{\text{(S)TFCS}}}\in\bigO(\tau^2)$.
\end{proof}

\section{Conclusion}
\label{sec:conclusion}
We investigated the parameterized computational complexity of the
problem of removing edges from a temporal graph to destroy all temporal cycles.
We showed NP-hardness even for temporal graphs with constant lifetime and
W[1]-hardness for the solution size parameter. On the positive side, our main
results are fixed-parameter tractability for the parameter ``number of
vertices'' and the treewidth of the underlying graph combined with the lifetime.

We conclude with some challenges for future research.
For the parameter lifetime~$\tau$, it remains open whether there exists a
polynomial-time algorithm for instances with $3 \leq \tau \leq 7$ in the strict
case and~$\tau=2$ in the non-strict case. We believe that, for the strict case,
our 3-SAT reduction %
can be modified to use only seven time-labels. %
Similarly to 
the work of \citet{ZschocheFMN18} in the context of temporal separators, we
could not resolve the question whether the non-strict variants are
fixed-parameter tractable for the combined parameter $\tau+k$, whereas for the
strict case, this is almost trivial.

We further leave as a future research challenge to investigate whether our
fixed-parameter tractability result for the parameter ``number of vertices'' can be improved: On the one hand, we would like to
improve the running time of the algorithm or show some conditional running time
lower bound to show that it likely cannot be improved significantly. On the
other hand, we leave open whether it is possible to obtain a polynomial-size problem kernel
for the number of vertices as a parameter.

Additionally, it seems natural to study (S)TFES and (S)TFCS variants restricted
to specific temporal graph classes (e.g., see \citet{Fluschnik19}). In
particular, we could not settle the parameterized complexity of our problem
variants when parameterized by (solely) the treewidth of the underlying graph.

Moreover, we remark that we focused on finding feedback edge sets, 
ignoring the presumably harder vertex variant;
however, one can observe that 
our W[1]-hardness result also transfers to the problem of finding feedback vertex sets
in temporal graphs.

Finally, we are interested in the possibilities to use the temporal feedback edge set size as a parameter for other temporal graph problems.
In order to design FPT-algorithms with this parameter it might be necessary to compute a solution to (S)TFES efficiently.
Since our hardness results refute efficient exact algorithms, we would like to obtain polynomial-time constant-factor approximation algorithms for our problems.

\paragraph*{Acknowledgments}
We would like to thank the anonymous reviewers for their careful checking of the manuscript and their valuable feedback.

\bibliography{strings-long,bibfile}

\end{document}